\documentclass[submission,copyright,creativecommons]{eptcs}
\providecommand{\event}{DCM 2023}
\usepackage{iftex}

\ifpdf
  \usepackage{underscore}
  \usepackage[T1]{fontenc}
\else
  \usepackage{breakurl}
\fi

\usepackage[ruled,vlined]{algorithm2e}
\usepackage{algorithmic}
\usepackage{amsfonts}
\usepackage{amsmath}
\usepackage{amsthm}
\usepackage{graphicx}
\usepackage{hyperref}
\usepackage{imakeidx}
\usepackage{multicol}
\usepackage[square, numbers]{natbib}
\usepackage{scalerel}
\usepackage{setspace}
\usepackage{tikz}
\usepackage{wrapfig}
\usepackage{doi}

\definecolor{graph_term}{RGB}{223, 223, 223}

\tikzset{ 
ev/.pic = { \fill[black] (0.0, 0.0) circle (1.0mm) ; } ,
iv/.pic = { \filldraw[black, fill=white, line width=0.5mm] circle (.8mm) ; } ,
de/.style = { black, line width=0.5mm, -> } ,
ue/.style = { black, line width=0.5mm } ,
he/.style = { black, line width=1.0mm } ,
hl/.style = { fill=white, draw=black, line width=0.5mm, rounded corners=0.1mm } ,
ht/.style = { fill=graph_term, draw=black, line width=0.25mm, rounded corners=0.1mm } ,
el/.style = { text width = 1.5mm, text height = 1.5mm, font = \scriptsize, align=center} ,
pn/.style = { text width = 10.0mm, align=center } ,
hn/.style = { text width = 3.0mm, align=center } ,
tn/.pic = { \draw[fill=white, draw=black, line width=0.5mm] circle (30.0mm) ; } ,
cn/.pic = { \draw[fill=white, draw=black, line width=0.5mm] circle (5.0mm) ; } ,
te/.style = { fill=white } ,
tb/.style = { draw=black, line width=0.25mm } ,
sl/.style = { sloped, scale=2, pos=0.5, allow upside down }
}

\newtheorem{theorem}{Theorem}[section]
\newtheorem{lemma}{Lemma}[section]
\theoremstyle{definition}

\title{Random Graph Generation in\\ Context-Free Graph Languages}
\author{Federico Vastarini
\institute{University of York\\ York, UK}
\email{federico.vastarini@york.ac.uk}
\and
Detlef Plump
\institute{University of York\\ York, UK}
\email{detlef.plump@york.ac.uk}
}
\newcommand{\titlerunning}{Random Graph Generation in Context-Free Graph Languages}
\newcommand{\authorrunning}{Federico Vastarini, Detlef Plump}
\hypersetup{
  bookmarksnumbered,
  pdftitle    = {\titlerunning},
  pdfauthor   = {\authorrunning},
  pdfsubject  = {Random Graph Generation in Context-Free Graph Languages},
  pdfkeywords = {hypergraph, hyperedge, replacement, language, grammar, mairson, uniform, sampling, generation, polynomial}
}
\begin{document}
\maketitle

\begin{abstract}
We present a method for generating random hypergraphs in context-free hypergraph languages. It is obtained by adapting Mairson's generation algorithm for context-free string grammars to the setting of hyperedge replacement grammars. Our main results are that for non-ambiguous hyperedge replacement grammars, the method generates hypergraphs uniformly at random and in quadratic time. We illustrate our approach by a running example of a hyperedge replacement grammar generating term graphs.\end{abstract}

\section{Introduction}
\label{sec:introduction}
We present a novel approach to the generation of random hypergraphs in user-specified domains. Our approach extends a method of Mairson for generating strings in context-free languages \cite{mairson-1994-gwi} to the setting of context-free hypergraph languages specified by hyperedge replacement grammars. Generating (or ``sampling") graphs and hypergraphs according to a given probability distributions is a problem that finds application in testing algorithms and programs working on graphs. Molecular biology and cryptography are two fields of potential application where our methods could find a concrete use besides the mere software testing. In \cite{kajino-2019-mhg} Kajino presents a novel approach for the representation of molecules through hypergraphs. Specifically adapting our method to this setting would provide an instrument for the exploration of new compounds in the field of molecular biology. The uniformity of the distribution of our method is a fundamental requirement for the development of cryptographic protocols. In \cite{goldreich-2011-cow} and \cite{micali-2002-tss} we may find some useful insights on how to model graph based algorithms in that domain. In the setting of hyperedge replacement grammars, we believe that there is an opportunity for the development of one-way functions.

Our generation algorithm uses as input a hyperedge replacement grammar in Chomsky normal form \cite{chomsky-1959-ocf} and a positive integer $n$. The former specifies the hypergraph language to sample from, the latter the size of the hypergraph to be generated. The algorithm then chooses a hypergraph at random from the slice of the language consisting of all members of size $n$. We show that if the grammar is non-ambiguous, the generated samples are uniformly distributed. The only requirements for our method are that the properties sought for the generated hypergraphs are representable by a hyperedge replacement language and that, to guarantee a uniform distribution, a non-ambiguous grammar is used as input.

We also show that our method generates a random hypergraph of size $n$ in time $O(n^2)$. This is the same time bound established by Mairson (for the first method) in the setting of random string generation in context-free languages.

\section{Hyperedge Replacement Grammars}
\label{sec:hrgl}

This section gives a concise overview of the definitions needed to understand the generation process. We also introduce our running example of a language of term graphs specified by hyperedge replacement. For comprehensive treatments of the theory of hyperedge replacement grammars and languages, we refer to Courcelle \cite{courcelle-1987-aad}, Drewes et al.\ \cite{drewes-1997-hrg} and Engelfriet \cite{engelfriet-1997-cfg}.

Let $\textit{type} \colon C \to \mathbb{N}_0$ be a typing function for a fixed set of labels $C$, then a \textit{hypergraph} over $C$ is a tuple $H=(V_H,E_H,\textit{att}_H,\textit{lab}_H,\textit{ext}_H)$ where $V_H$ is a finite set of \textit{vertices}, $E_H$ is a finite set of \textit{hyperedges}, $\textit{att}_H\colon E_H \to V^*_H$ is a mapping assigning a sequence of \textit{attachment nodes} to each $e \in E_H$, $\textit{lab}_H\colon E_H \to C$ is a function that maps each hyperedge to a \textit{label} such that $\textit{type}(\textit{lab}_{H}(e)) = |\textit{att}_{H}(e)|$, $\textit{ext}_H \in V^*_H$ is a sequence of pairwise distinct \textit{external nodes} (Figure \ref{fig:hypergraph}).
The class of all hypergraphs over $C$ is denoted by $\mathcal{H}_C$.

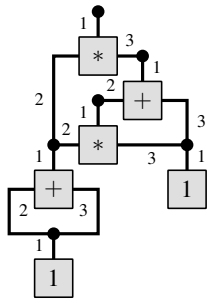
\begin{figure}[htpb]
\centering
\resizebox{2.75cm}{!}{%
\begin{tikzpicture}
 \begin{scope} [local bounding box = c1, shift = {(0.0, 0.0)}]
  \draw [ue] (-0.7, -0.4) -- (-0.7, 0.0) node[el, midway, left] {1} ;
  \draw [ue] (-1.0, -0.7) -- (-1.4, -0.7) -- (-1.4, -1.4) node[el, midway, right] {2} -- (-0.7, -1.4) ;
  \draw [ue] (-0.4, -0.7) -- (0.0, -0.7) -- (0.0, -1.4) node[el, midway, left] {3} -- (-0.7, -1.4) ;
  \draw [ue] (0.0, 1.7) -- (0.0, 2.1) node[el, midway, left] {1} ;
  \draw [ue] (-0.3, 1.4) -- (-0.7, 1.4) -- (-0.7, 0.0) node[el, midway, left] {2} ;
  \draw [ue] (0.3, 1.4) -- (0.7, 1.4) node[el, midway, above] {3} ;
  \draw [ue] (0.0, 0.3) -- (0.0, 0.7) node[el, midway, left] {1} ;
  \draw [ue] (-0.3, 0.0) -- (-0.7, 0.0) node[el, midway, above] {2} ;
  \draw [ue] (0.3, 0.0) -- (1.4, 0.0) node[el, midway, below] {3} ;
  \draw [ue] (0.7, 1.0) -- (0.7, 1.4) node[el, midway, right] {1} ;
  \draw [ue] (0.4, 0.7) -- (0.0, 0.7) node[el, midway, above] {2} ;
  \draw [ue] (1.0, 0.7) -- (1.4, 0.7) -- (1.4, 0.0) node[el, midway, right] {3} ;
  \draw [ue] (1.4, -0.4) -- (1.4, 0.0) node[el, midway, right] {1} ;
  \draw [ue] (-0.7, -1.8) -- (-0.7, -1.4) node[el, midway, left] {1} ;
  \draw [ht] (-1.0, -0.4) rectangle (-0.4, -1.0) node[midway] {$+$} ;
  \draw [ht] (-0.3, 0.3) rectangle (0.3, -0.3) node[midway] {$*$} ;
  \draw [ht] (1.1, -0.4) rectangle (1.7, -1.0) node[midway] {$1$} ;
  \draw [ht] (-0.3, 1.7) rectangle (0.3, 1.1) node[midway] {$*$} ;
  \draw [ht] (0.4, 1.0) rectangle (1.0, 0.4) node[midway] {$+$} ;
  \draw [ht] (-1.0, -1.8) rectangle (-0.4, -2.4) node[midway] {$1$} ;
  \path (0.0, 2.1) pic{ev} ;
  \path (-0.7, 0.0) pic{ev} ;
  \path (1.4, 0.0) pic{ev} ;
  \path (0.7, 1.4) pic{ev} ;
  \path (0.0, 0.7) pic{ev} ;
  \path (-0.7, -1.4) pic{ev} ;
 \end{scope}
\end{tikzpicture}%
}
\caption{A term graph}
\label{fig:hypergraph}
\end{figure}

We write $\textit{type}(H)$ for $|\textit{ext}_{H}|$ and call $H$ an $n$-hypergraph if $\textit{type}(H) = n$. The length of the sequence of \textit{attachments} $|\textit{att}_{H}(e)|$ is the \textit{type} of $e$. Hyperedge $e$ is an $m$-hyperedge if $\textit{type}(\textit{lab}(e)) = m$. We also write $\textit{type}(e) = m$ if the context is clear. If an $n$-hypergraph has exactly $1$ hyperedge and all its nodes are external, that is $E_H = \{e\}$ and $|V_H| = n$, it is called the \textit{handle} induced by $e$ and denoted by $I_e$. Moreover if $\textit{type}(e) = n$ and $\textit{ext}_H = \textit{att}(e)$ such a hypergraph is called the handle induced by $A$ and denoted by $A^{\bullet}$. We write $\textit{att}_{H}(e)_i$ for the $i$th attachment node of $e \in E_H$ and $\textit{ext}_{H,i}$ for the $i$th external node of $H$. The set $E^X_H = \{e \in E \mid \textit{lab}_H(e) \in X\}$ is the subset of $E_H$ with labels in $X \subseteq C$. We define $|H| = |V_H| + |E_H|$ as the \textit{size} of $H$ and we call $H$ a size-$n$-hypergraph if $|H| = n$.

Figure \ref{fig:hypergraph} shows a \emph{term graph}, a form of acyclic hypergraphs that represent functional expressions with possibly shared subexpressions. (See \cite{plump-1998-tgr} for an introduction to the area of term graph rewriting.) Grey boxes represent hyperedges labelled with the function symbols $*$, $+$ and $1$, while nodes are drawn as black bullets. Lines connect hyperedges with their attachment nodes, whose position in the attachment sequence is given by small numbers.

Two hypergraphs $H, H' \in \mathcal{H}_C$ are \textit{isomorphic}, denoted $H \cong H'$, if there are bijective mappings $h_V\colon V_H \to V_{H'}$ and $h_E\colon E_H \to E_{H'}$ such that $h^*_{V}(\textit{att}_{H}(e)) = \textit{att}_{H'}(h_E(e))$ and $\textit{lab}_{H}(e) = \textit{lab}_{H'}(h_E(e))$ for each $e \in E_H$ and $h^*_{V}(\textit{ext}_{H}) = \textit{ext}_{H'}$. Two isomorphic hypergraphs are considered to be the same. A hypergraph $H$ is a subgraph of $H'$, denoted as $H \subseteq H'$ if $V_H \subseteq V_{H'}$ and $E_H \subseteq E_{H'}$.

Hypergraphs are generated by replacement operations. Let $H \in \mathcal{H}_C$ and $B \subseteq E_H$ and let $\textit{repl} \colon B \to \mathcal{H}_C$ be a mapping with $\textit{type}(\textit{repl}(e)) = \textit{type}(e)$ for each $e \in B$. Then the \textit{replacement} of the hyperedges in $B$ with respect to $\textit{repl}(e)$ is defined by the operations:
remove the subset $B$ of hyperedges from $E_H$; for each $e \in B$, disjointly add the vertices and the hyperedges of $\textit{repl}(e)$; for each $e \in B$ and $1 \leq i \leq \textit{type}(e)$, fuse the $i$th external node $\textit{ext}_{\textit{repl}(e),i}$ with the $i$th attachment node $\textit{att}_{B}(e)_i$.

We denote the resulting hypergraph by $H[e_1 / R_1, \dots ,e_n / R_n]$, where $B = \{e_1, \dots ,e_n\}$ and $\textit{repl}(e_i) = R_i$ for $1 \leq i \leq n$, or $H[\textit{repl}]$. The replacement preserves the external nodes, thus $ext_{H[\textit{repl}]} = ext_H$.

Given the subsets $\Sigma, N \subseteq C$ used as terminal and non-terminal labels, with $\Sigma \cap N = \emptyset$, we denote $E^{\Sigma}_H$ and $E^N_H$ respectively the subsets of terminal and non-terminal hyperedges of $H$.

The replacements applied during the generation are defined in productions: $p = (A,R)$ is a \textit{production} over $N$, where $\textit{lhs}(p) = A \in N$ is the label of the replaced hyperedge and $\textit{rhs}(p) = R \in \mathcal{H}_C$ is a hypergraph with $\textit{type}(R) = \textit{type}(A)$.
If $|ext_R| = |V_R|$ and $E_R = \emptyset$, then $p$ is said to be \textit{empty}.

Let $H \in \mathcal{H}_C$ and let $p = (\textit{lab}(e),R)$, with $e \in E_H$, then a \textit{direct derivation} $H \Rightarrow_p H'$ is obtained by the replacement $H' = H[e/R]$.

A sequence $d$ of direct derivations $H_0 \Rightarrow_{p_1} \dots \Rightarrow_{p_k} H_k$ of length $k$ with $(p_1, \dots ,p_k) \in P$ is denoted as $H \Rightarrow^k H_k$ or $H \Rightarrow^*_P H_k$ if the length is not relevant. We denote it as $H \Rightarrow^* H_k$ if the sequence is clear from the context.

A derivation $H \Rightarrow^* H'$ of length $0$ is given if $H \cong H'$.

Given an ordered set $\{\alpha_1, \ldots, \alpha_n\}$ where $a_i < a_j$ if $i < j \in \mathbb{N}$ we define a \textit{hyperedge replacement grammar}, or \textit{HRG} as a tuple $G = (N,\Sigma,P,S,(mark_{p})_{p \in P})$ where $N \subseteq C$ is a finite set of non-terminal labels, $\Sigma \subseteq C$ is a finite set of terminal labels with $N \cap \Sigma = \emptyset$, $P$ is a finite set of productions, $S \in N$ is the starting symbol, $(mark_{p})_{p \in P}$ is a family of functions $mark_p: E_R \to \{\alpha_1, \ldots, \alpha_n\}$ assigning a mark to each hyperedge in the right-hand side of a production $p$ (Figure \ref{fig:tggrammar}). For each pair $e_i, e_j \in E_R$ with $i \neq j$, $\textit{mark}(e_i) \neq \textit{mark}(e_j)$.

We denote as $P^A \subseteq P$ the subset of productions where $\textit{lhs}(p) = A$.
We call a production $p = (A,R) \in P_N \subseteq P$ \textit{non-terminal} if $E^N_R \neq \emptyset$ or \textit{terminal} if $p = (A,R) \in P_{\Sigma} = P \backslash P_N$.

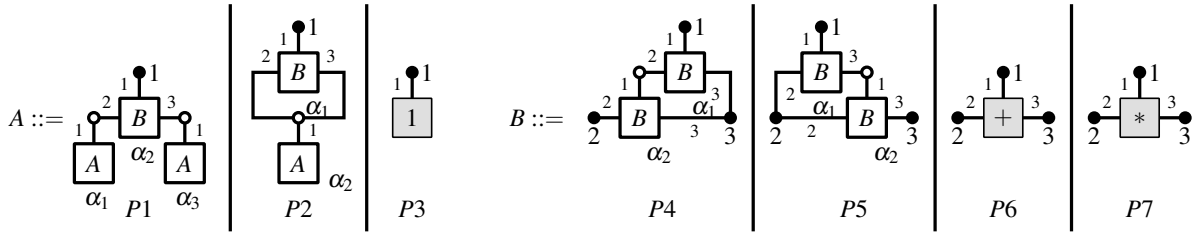
\begin{figure}[htpb]
\centering
\resizebox{\textwidth}{!}{%
\begin{tikzpicture}
 \node[text width = 1.0cm, align=center] at (-0.175, 0.0) {$A::=$} ;
 
 \begin{scope} [local bounding box = c1, shift = {(1.4, 0.0)}]
  \draw [ue] (0.0, 0.3) -- (0.0, 0.7) node[el, midway, left] {1} ;
  \draw [ue] (-0.3, 0.0) -- (-0.7, 0.0) node[el, midway, above] {2} ;
  \draw [ue] (0.3, 0.0) -- (0.7, 0.0) node[el, midway, above] {3} ;
  \draw [ue] (-0.7, -0.4) -- (-0.7, 0.0) node[el, midway, left] {1} ;
  \draw [ue] (0.7, -0.4) -- (0.7, 0.0) node[el, midway, right] {1} ;
  \draw [hl] (0.3, 0.3) rectangle (-0.3, -0.3) node[midway] {$B$} node[below right] {$\alpha_2$} ;
  \draw [hl] (1.0, -0.4) rectangle (0.4, -1.0) node[midway] {$A$} node[below right] {$\alpha_3$} ;
  \draw [hl] (-0.4, -0.4) rectangle (-1.0, -1.0) node[midway] {$A$} node[below right] {$\alpha_1$} ;
  \path (0.0, 0.7) pic{ev} node[right] {1} ;
  \path (-0.7, 0.0) pic{iv} ;
  \path (0.7, 0.0) pic{iv} ;
  \node [pn] at (0.0, -1.4) {$P1$} ;
 \end{scope}
 
 \draw[black, line width=0.5mm] (2.8, 1.75) -- (2.8, -1.75) ;
 
 \begin{scope} [local bounding box = c1, shift = {(3.85, 0.0)}]
  \draw [ue] (0.0, 1.0) -- (0.0, 1.4) node[el, midway, left] {1} ;
  \draw [ue] (-0.3, 0.7) -- (-0.7, 0.7) node[el, midway, above] {2} -- (-0.7, 0.0) -- (0.0, 0.0) ;
  \draw [ue] (0.3, 0.7) -- (0.7, 0.7) node[el, midway, above] {3} -- (0.7, 0.0) -- (0.0, 0.0) ;
  \draw [ue] (0.0, -0.4) -- (0.0, 0.0) node[el, midway, right] {1} ;
  \draw [hl] (-0.3, 1.0) rectangle (0.3, 0.4) node[midway] {$B$} node[below] {$\alpha_1$} ;
  \draw [hl] (-0.3, -0.4) rectangle (0.3, -1.0) node[midway] {$A$} node[right] {$\alpha_2$} ;
  \path (0.0, 1.4) pic{ev} node[right] {1} ;
  \path (0.0, 0.0) pic{iv} ;
  \node [pn] at (0.0, -1.4) {$P2$} ;
 \end{scope}
 
 \draw[black, line width=0.5mm] (4.9, 1.75) -- (4.9, -1.75) ;
 
 \begin{scope} [local bounding box = c1, shift = {(5.6, 0.0)}]
  \draw [ue] (0.0, 0.3) -- (0.0, 0.7) node[el, midway, left] {1} ;
  \draw [ht] (-0.3, 0.3) rectangle (0.3, -0.3) node[midway] {$1$} ;
  \path (0.0, 0.7) pic{ev} node[right] {1} ;
  \node [pn] at (0.0, -1.4) {$P3$} ;
 \end{scope}
 
 \node[text width = 1.0cm, align=center] at (7.525, 0.0) {$B::=$} ;
 
 \begin{scope} [local bounding box = c1, shift = {(9.8, 0.0)}]
  \draw [ue] (0.0, 1.0) -- (0.0, 1.4) node[el, midway, left] {1} ;
  \draw [ue] (-0.3, 0.7) -- (-0.7, 0.7) node[el, midway, above] {2} ;
  \draw [ue] (0.3, 0.7) -- (0.7, 0.7) -- (0.7, 0.0) node[el, midway, left] {3} ;
  \draw [ue] (-0.7, 0.3) -- (-0.7, 0.7) node[el, midway, left] {1} ;
  \draw [ue] (-1.0, 0.0) -- (-1.4, 0.0) node[el, midway, above] {2} ;
  \draw [ue] (-0.4, 0.0) -- (0.7, 0.0) node[el, midway, below] {3} ;
  \draw [hl] (-0.3, 1.0) rectangle (0.3, 0.4) node[midway] {$B$} node[below] {$\alpha_1$} ;
  \draw [hl] (-1.0, 0.3) rectangle (-0.4, -0.3) node[midway] {$B$} node[below] {$\alpha_2$} ;
  \path (0.0, 1.4) pic{ev} node[right] {1} ;
  \path (-1.4, 0.0) pic{ev} node[below] {2} ;
  \path (0.7, 0.0) pic{ev} node[below] {3} ;
  \path (-0.7, 0.7) pic{iv} ;
  \node [pn] at (-0.35, -1.4) {$P4$} ;
 \end{scope}
 
 \draw[black, line width=0.5mm] (10.85, 1.75) -- (10.85, -1.75) ;
 
 \begin{scope} [local bounding box = c1, shift = {(11.9, 0.0)}]
  \draw [ue] (0.0, 1.0) -- (0.0, 1.4) node[el, midway, left] {1} ;
  \draw [ue] (-0.3, 0.7) -- (-0.7, 0.7) -- (-0.7, 0.0) node[el, midway, right] {2} ;
  \draw [ue] (0.3, 0.7) -- (0.7, 0.7) node[el, midway, above] {3} ;
  \draw [ue] (0.7, 0.3) -- (0.7, 0.7) node[el, midway, right] {1} ;
  \draw [ue] (0.4, 0.0) -- (-0.7, 0.0) node[el, midway, below] {2} ;
  \draw [ue] (1.0, 0.0) -- (1.4, 0.0) node[el, midway, above] {3} ;
  \draw [hl] (0.3, 1.0) rectangle (-0.3, 0.4) node[midway] {$B$} node[below right] {$\alpha_1$} ;
  \draw [hl] (0.4, 0.3) rectangle (1.0, -0.3) node[midway] {$B$} node[below] {$\alpha_2$} ;
  \path (0.0, 1.4) pic{ev} node[right] {1} ;
  \path (-0.7, 0.0) pic{ev} node[below] {2} ;
  \path (1.4, 0.0) pic{ev} node[below] {3} ;
  \path (0.7, 0.7) pic{iv} ;
  \node [pn] at (0.5, -1.4) {$P5$} ;
 \end{scope}
 
 \draw[black, line width=0.5mm] (13.65, 1.75) -- (13.65, -1.75) ;
 
 \begin{scope} [local bounding box = c1, shift = {(14.7, 0.0)}]
  \draw [ue] (0.0, 0.3) -- (0.0, 0.7) node[el, midway, left] {1} ;
  \draw [ue] (-0.3, 0.0) -- (-0.7, 0.0) node[el, midway, above] {2} ;
  \draw [ue] (0.3, 0.0) -- (0.7, 0.0) node[el, midway, above] {3} ;
  \draw [ht] (-0.3, 0.3) rectangle (0.3, -0.3) node[midway] {$+$} ;
  \path (0.0, 0.7) pic{ev} node[right] {1} ;
  \path (-0.7, 0.0) pic{ev} node[below] {2} ;
  \path (0.7, 0.0) pic{ev} node[below] {3} ;
  \node [pn] at (0.0, -1.4) {$P6$} ;
 \end{scope}
 
 \draw[black, line width=0.5mm] (15.75, 1.75) -- (15.75, -1.75) ;
 
 \begin{scope} [local bounding box = c1, shift = {(16.8, 0.0)}]
  \draw [ue] (0.0, 0.3) -- (0.0, 0.7) node[el, midway, left] {1} ;
  \draw [ue] (-0.3, 0.0) -- (-0.7, 0.0) node[el, midway, above] {2} ;
  \draw [ue] (0.3, 0.0) -- (0.7, 0.0) node[el, midway, above] {3} ;
  \draw [ht] (-0.3, 0.3) rectangle (0.3, -0.3) node[midway] {$*$} ;
  \path (0.0, 0.7) pic{ev} node[right] {1} ;
  \path (-0.7, 0.0) pic{ev} node[below] {2} ;
  \path (0.7, 0.0) pic{ev} node[below] {3} ;
  \node [pn] at (0.0, -1.4) {$P7$} ;
 \end{scope}
\end{tikzpicture}%
}
\caption{An ambiguous hyperedge replacement grammar for term graphs}
\label{fig:tggrammar}
\end{figure}

The marking of the hyperedges in the $\textit{rhs}$ of each production, represents the order in which the replacements are carried out $(\alpha_1, \alpha_2, \ldots, \alpha_{n-1}, \alpha_n)$. It allows for the definitions of ordered derivation tree and leftmost derivation.

Given a set of productions $P$, we denote by $T_P$ the set of all ordered trees over $P$ which is inductively defined as follows: for each $p \in P$, $p \in T_P$; for $t_1, \ldots, t_n \in T_P$ and $p \in P$, $p(t_1, \ldots, t_n) \in T_P$.

Then, given an $\textit{HRG}$ $G$, an \textit{ordered derivation tree} $t$ for $e$ such that $\textit{lab}(e) = X \in N$, is a tree $p(t_{\alpha_1}, \ldots, t_{\alpha_n})$ in $T_P$, such that $p = (X, R)$ is a production in $P$, and $t_{\alpha_1}, \ldots, t_{\alpha_n}$ are derivation trees for $e_1 \ldots e_n$, such that $X_1 \ldots X_n$ are the labels of the non-terminal hyperedges in $R$ marked with $\alpha_1 \ldots \alpha_n$, respectively (Figure \ref{fig:tree}).

We define the \textit{yield} of an ordered derivation tree $t$, denoted with $yield(t)$, as the sequence of replacements: $\textit{yield}(p(t_{\alpha_1},\ldots,t_{\alpha_n}))=\textit{rhs}(p)[e_1 / \textit{yield}(t_{\alpha_1}),\ldots,e_n / \textit{yield}(t_{\alpha_n})]$.

Let $t$ be an ordered derivation tree for a hypergraph $H$ obtained from a derivation $d = S^{\bullet} \Rightarrow^*_P H$ and $\textit{trav}(t)$ its pre-ordered visit. Then $d$ is said to be a \textit{leftmost derivation}, denoted as $\textit{lmd}(H)$, if and only if the order of the applied productions of $d$ corresponds to $\textit{trav}(t)$.

Since we need a measure to ensure the termination of the proposed algorithm, we define an $\textit{HRG}$ to be \textit{non-contracting} if for each direct derivation $H \Rightarrow_p H'$, $|H| \leq |H'|$. We call a grammar \textit{essentially non-contracting} if there exists $p = (S,R) \in P$ such that $p$ is the empty production.

The \textit{hyperedge replacement language} ($\textit{HRL}$) generated by an $\textit{HRG}$ is the set $L(G) = \{H \in \mathcal{H}_{\Sigma} \mid S^{\bullet} \Rightarrow^*_P H\}$. We define for each $A \in N$, $L^{A}(G) = \{H \in \mathcal{H}_{\Sigma} \mid A^{\bullet} \Rightarrow^*_P H\}$. We also define for $n \in \mathbb{N}$, $L^{A}_{n}(G) = \{H \in \mathcal{H}_{\Sigma} \mid A^{\bullet} \Rightarrow^*_P H \wedge |H| = n\}$. Clearly $L^{A}_{n}(G) \subseteq L^{A}(G)$.
We denote as $|L^A_n|$ the size of the set of all size-$n$-hypergraphs in $L$ that can be derived from $A^{\bullet}$.

For example, the hyperedge replacement grammar in Figure \ref{fig:tggrammar} generates the class of all term graphs with function symbols in $\{*, +, 1\}$. Note that hyperedges with non-terminal labels are depicted as white boxes. A derivation with the Chomsky normal form version of this grammar is given in Figure \ref{fig:termsgraphderivation}.

We define a grammar $G$ to be \textit{ambiguous} if there are ordered derivation trees $t_1, t_2 \in T_P$, such that $t_1 \neq t_2$ and $yield(t_1) \cong yield(t_2)$, or equivalently, if there exist $H, H' \in L(G)$ such that $H \cong H'$ and $\textit{lmd}(H) \neq \textit{lmd}(H')$. If $yield(t_1),yield(t_2) \in L_n(G)$ we say that $G$ is \textit{n-ambiguous}. A non-ambiguous version of the term graphs grammar is given in Figure \ref{fig:natggrammar}.

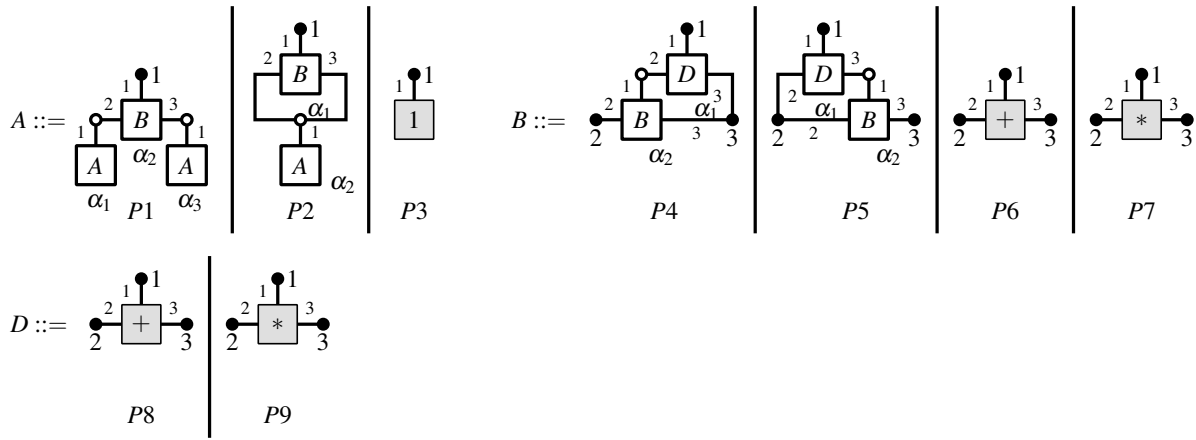
\begin{figure}[htpb]
\centering
\resizebox{\textwidth}{!}{%
\begin{tikzpicture}
 \node[text width = 1.0cm, align=center] at (-0.175, 0.0) {$A::=$} ;
 
 \begin{scope} [local bounding box = c1, shift = {(1.4, 0.0)}]
  \draw [ue] (0.0, 0.3) -- (0.0, 0.7) node[el, midway, left] {1} ;
  \draw [ue] (-0.3, 0.0) -- (-0.7, 0.0) node[el, midway, above] {2} ;
  \draw [ue] (0.3, 0.0) -- (0.7, 0.0) node[el, midway, above] {3} ;
  \draw [ue] (-0.7, -0.4) -- (-0.7, 0.0) node[el, midway, left] {1} ;
  \draw [ue] (0.7, -0.4) -- (0.7, 0.0) node[el, midway, right] {1} ;
  \draw [hl] (0.3, 0.3) rectangle (-0.3, -0.3) node[midway] {$B$} node[below right] {$\alpha_2$} ;
  \draw [hl] (1.0, -0.4) rectangle (0.4, -1.0) node[midway] {$A$} node[below right] {$\alpha_3$} ;
  \draw [hl] (-0.4, -0.4) rectangle (-1.0, -1.0) node[midway] {$A$} node[below right] {$\alpha_1$} ;
  \path (0.0, 0.7) pic{ev} node[right] {1} ;
  \path (-0.7, 0.0) pic{iv} ;
  \path (0.7, 0.0) pic{iv} ;
  \node [pn] at (0.0, -1.4) {$P1$} ;
 \end{scope}
 
 \draw[black, line width=0.5mm] (2.8, 1.75) -- (2.8, -1.75) ;
 
 \begin{scope} [local bounding box = c1, shift = {(3.85, 0.0)}]
  \draw [ue] (0.0, 1.0) -- (0.0, 1.4) node[el, midway, left] {1} ;
  \draw [ue] (-0.3, 0.7) -- (-0.7, 0.7) node[el, midway, above] {2} -- (-0.7, 0.0) -- (0.0, 0.0) ;
  \draw [ue] (0.3, 0.7) -- (0.7, 0.7) node[el, midway, above] {3} -- (0.7, 0.0) -- (0.0, 0.0) ;
  \draw [ue] (0.0, -0.4) -- (0.0, 0.0) node[el, midway, right] {1} ;
  \draw [hl] (-0.3, 1.0) rectangle (0.3, 0.4) node[midway] {$B$} node[below] {$\alpha_1$} ;
  \draw [hl] (-0.3, -0.4) rectangle (0.3, -1.0) node[midway] {$A$} node[right] {$\alpha_2$} ;
  \path (0.0, 1.4) pic{ev} node[right] {1} ;
  \path (0.0, 0.0) pic{iv} ;
  \node [pn] at (0.0, -1.4) {$P2$} ;
 \end{scope}
 
 \draw[black, line width=0.5mm] (4.9, 1.75) -- (4.9, -1.75) ;
 
 \begin{scope} [local bounding box = c1, shift = {(5.6, 0.0)}]
  \draw [ue] (0.0, 0.3) -- (0.0, 0.7) node[el, midway, left] {1} ;
  \draw [ht] (-0.3, 0.3) rectangle (0.3, -0.3) node[midway] {$1$} ;
  \path (0.0, 0.7) pic{ev} node[right] {1} ;
  \node [pn] at (0.0, -1.4) {$P3$} ;
 \end{scope}
 
 \node[text width = 1.0cm, align=center] at (7.525, 0.0) {$B::=$} ;
 
 \begin{scope} [local bounding box = c1, shift = {(9.8, 0.0)}]
  \draw [ue] (0.0, 1.0) -- (0.0, 1.4) node[el, midway, left] {1} ;
  \draw [ue] (-0.3, 0.7) -- (-0.7, 0.7) node[el, midway, above] {2} ;
  \draw [ue] (0.3, 0.7) -- (0.7, 0.7) -- (0.7, 0.0) node[el, midway, left] {3} ;
  \draw [ue] (-0.7, 0.3) -- (-0.7, 0.7) node[el, midway, left] {1} ;
  \draw [ue] (-1.0, 0.0) -- (-1.4, 0.0) node[el, midway, above] {2} ;
  \draw [ue] (-0.4, 0.0) -- (0.7, 0.0) node[el, midway, below] {3} ;
  \draw [hl] (-0.3, 1.0) rectangle (0.3, 0.4) node[midway] {$D$} node[below] {$\alpha_1$} ;
  \draw [hl] (-1.0, 0.3) rectangle (-0.4, -0.3) node[midway] {$B$} node[below] {$\alpha_2$} ;
  \path (0.0, 1.4) pic{ev} node[right] {1} ;
  \path (-1.4, 0.0) pic{ev} node[below] {2} ;
  \path (0.7, 0.0) pic{ev} node[below] {3} ;
  \path (-0.7, 0.7) pic{iv} ;
  \node [pn] at (-0.35, -1.4) {$P4$} ;
 \end{scope}
 
 \draw[black, line width=0.5mm] (10.85, 1.75) -- (10.85, -1.75) ;
 
 \begin{scope} [local bounding box = c1, shift = {(11.9, 0.0)}]
  \draw [ue] (0.0, 1.0) -- (0.0, 1.4) node[el, midway, left] {1} ;
  \draw [ue] (-0.3, 0.7) -- (-0.7, 0.7) -- (-0.7, 0.0) node[el, midway, right] {2} ;
  \draw [ue] (0.3, 0.7) -- (0.7, 0.7) node[el, midway, above] {3} ;
  \draw [ue] (0.7, 0.3) -- (0.7, 0.7) node[el, midway, right] {1} ;
  \draw [ue] (0.4, 0.0) -- (-0.7, 0.0) node[el, midway, below] {2} ;
  \draw [ue] (1.0, 0.0) -- (1.4, 0.0) node[el, midway, above] {3} ;
  \draw [hl] (0.3, 1.0) rectangle (-0.3, 0.4) node[midway] {$D$} node[below right] {$\alpha_1$} ;
  \draw [hl] (0.4, 0.3) rectangle (1.0, -0.3) node[midway] {$B$} node[below] {$\alpha_2$} ;
  \path (0.0, 1.4) pic{ev} node[right] {1} ;
  \path (-0.7, 0.0) pic{ev} node[below] {2} ;
  \path (1.4, 0.0) pic{ev} node[below] {3} ;
  \path (0.7, 0.7) pic{iv} ;
  \node [pn] at (0.5, -1.4) {$P5$} ;
 \end{scope}
 
 \draw[black, line width=0.5mm] (13.65, 1.75) -- (13.65, -1.75) ;
 
 \begin{scope} [local bounding box = c1, shift = {(14.7, 0.0)}]
  \draw [ue] (0.0, 0.3) -- (0.0, 0.7) node[el, midway, left] {1} ;
  \draw [ue] (-0.3, 0.0) -- (-0.7, 0.0) node[el, midway, above] {2} ;
  \draw [ue] (0.3, 0.0) -- (0.7, 0.0) node[el, midway, above] {3} ;
  \draw [ht] (-0.3, 0.3) rectangle (0.3, -0.3) node[midway] {$+$} ;
  \path (0.0, 0.7) pic{ev} node[right] {1} ;
  \path (-0.7, 0.0) pic{ev} node[below] {2} ;
  \path (0.7, 0.0) pic{ev} node[below] {3} ;
  \node [pn] at (0.0, -1.4) {$P6$} ;
 \end{scope}
 
 \draw[black, line width=0.5mm] (15.75, 1.75) -- (15.75, -1.75) ;
 
 \begin{scope} [local bounding box = c1, shift = {(16.8, 0.0)}]
  \draw [ue] (0.0, 0.3) -- (0.0, 0.7) node[el, midway, left] {1} ;
  \draw [ue] (-0.3, 0.0) -- (-0.7, 0.0) node[el, midway, above] {2} ;
  \draw [ue] (0.3, 0.0) -- (0.7, 0.0) node[el, midway, above] {3} ;
  \draw [ht] (-0.3, 0.3) rectangle (0.3, -0.3) node[midway] {$*$} ;
  \path (0.0, 0.7) pic{ev} node[right] {1} ;
  \path (-0.7, 0.0) pic{ev} node[below] {2} ;
  \path (0.7, 0.0) pic{ev} node[below] {3} ;
  \node [pn] at (0.0, -1.4) {$P7$} ;
 \end{scope}
 
 \node[text width = 1.0cm, align=center] at (-0.175, -3.15) {$D::=$} ;
 
 \begin{scope} [local bounding box = c1, shift = {(1.4, -3.15)}]
  \draw [ue] (0.0, 0.3) -- (0.0, 0.7) node[el, midway, left] {1} ;
  \draw [ue] (-0.3, 0.0) -- (-0.7, 0.0) node[el, midway, above] {2} ;
  \draw [ue] (0.3, 0.0) -- (0.7, 0.0) node[el, midway, above] {3} ;
  \draw [ht] (-0.3, 0.3) rectangle (0.3, -0.3) node[midway] {$+$} ;
  \path (0.0, 0.7) pic{ev} node[right] {1} ;
  \path (-0.7, 0.0) pic{ev} node[below] {2} ;
  \path (0.7, 0.0) pic{ev} node[below] {3} ;
  \node [pn] at (0.0, -1.4) {$P8$} ;
 \end{scope}
 
 \draw[black, line width=0.5mm] (2.45, -2.1) -- (2.45, -4.9) ;
 
 \begin{scope} [local bounding box = c1, shift = {(3.5, -3.15)}]
  \draw [ue] (0.0, 0.3) -- (0.0, 0.7) node[el, midway, left] {1} ;
  \draw [ue] (-0.3, 0.0) -- (-0.7, 0.0) node[el, midway, above] {2} ;
  \draw [ue] (0.3, 0.0) -- (0.7, 0.0) node[el, midway, above] {3} ;
  \draw [ht] (-0.3, 0.3) rectangle (0.3, -0.3) node[midway] {$*$} ;
  \path (0.0, 0.7) pic{ev} node[right] {1} ;
  \path (-0.7, 0.0) pic{ev} node[below] {2} ;
  \path (0.7, 0.0) pic{ev} node[below] {3} ;
  \node [pn] at (0.0, -1.4) {$P9$} ;
 \end{scope}
\end{tikzpicture}%
}
\caption{A non-ambiguous hyperedge replacement grammar for term graphs}
\label{fig:natggrammar}
\end{figure}

\section{Random hypergraph generation}
\label{sec:randomgeneration}

In 1994, Mairson proposed a pair of methods for the sampling of strings from context-free grammars \cite{mairson-1994-gwi} . His approach requires, as input, a grammar $G$ in Chomsky normal form and the length $n$ of the word to be generated. He proves that, if $G$ is non-ambiguous, such a word is generated uniformly at random. The first method has a time complexity of $O(n^2)$ while requiring $O(n)$ space, and vice versa, the second method runs in linear time using quadratic space. In the following we adapt the first of Mairson's methods to hyperedge replacement grammars. We use our running example of a term graph language to illustrate the generation process.

We define a \textit{Chomsky normal form} ($\textit{CNF}$) for hyperedge replacement grammars as a tuple $G_{\textit{CNF}} = (N,\Sigma,P,S,(mark_{p})_{p \in P})$ where:\begin{itemize}
\setlength\itemsep{-0.2em}
\item $N \subseteq C$ is a finite set of non-terminal labels
\item $\Sigma \subseteq C$ is a finite set of terminal labels with $N \cap \Sigma = \emptyset$
\item $P$ is a finite set of productions
\item $S \in N$ is the starting symbol
\item $(mark_{p})_{p \in P}$ is a family of functions $mark_p: E_R \to \{\alpha, \beta\}$ assigning a mark to each hyperedge in the right-hand side of a production $p$
\end{itemize}
Each production $p = (A,R) \in P$ satisfies one of the following constraints:
\begin{itemize}
\setlength\itemsep{-0.2em}
\item $E_R = \{e_1, e_2\}$ where $\textit{lab}(e_1), \textit{lab}(e_2) \in N$ and $\textit{mark}(e_1) \neq \textit{mark}(e_2)$, in which case the replacement is firstly carried out on the hyperedge marked with $\alpha$, then on the one marked with $\beta$
\item $E_R = \{e_1\}$ where $lab(e_1) \in \Sigma$ and $mark(e_1) = \alpha$
\item $E_R = \emptyset$, $|V_R| > |\textit{ext}_R|$
\item $A = S$, $p$ is the empty production and for each $q \in P$, for each $e \in \textit{rhs}(q)$, $\textit{lab}(e) \neq S$
\end{itemize}
Note that in the first two cases, $\textit{rhs}(p)$ contains either exactly two non-terminal hyperedges or a single terminal hyperedge and may also contain isolated nodes. Productions according to the third case are considered as \textit{terminal} productions. The last case specifies that the empty production is only allowed if there is no other production having the starting symbol in its right-hand side. The grammar in Figure \ref{fig:cnftggrammar} is the $\textit{CNF}$ version of the term graph grammar in Figure \ref{fig:tggrammar}.

\begin{figure}[htpb]
\centering
\resizebox{15cm}{!}{%
\begin{tikzpicture}
 \node[text width = 1.0cm, align=center] at (-0.175, 0.0) {$A::=$} ;
 
 \begin{scope} [local bounding box = c1, shift = {(1.4, 0.0)}]
  \draw [ue] (0.0, 0.3) -- (0.0, 0.7) node[el, midway, left] {1} ;
  \draw [ue] (-0.3, 0.0) -- (-0.7, 0.0) node[el, midway, above] {2} ;
  \draw [ue] (0.3, 0.0) -- (0.7, 0.0) node[el, midway, above] {3} ;
  \draw [ue] (-0.7, -0.4) -- (-0.7, 0.0) node[el, midway, left] {1} ;
  \draw [hl] (-0.3, 0.3) rectangle (0.3, -0.3) node[midway] {$C$} node[below] {$\alpha$} ;
  \draw [hl] (-1.0, -0.4) rectangle (-0.4, -1.0) node[midway] {$A$} node[right] {$\beta$} ;
  \path (0.0, 0.7) pic{ev} node[right] {1} ;
  \path (-0.7, 0.0) pic{iv} ;
  \path (0.7, 0.0) pic{iv} ;
  \node [pn] at (-0.35, -1.4) {$P1$} ;
 \end{scope}
 
 \draw[black, line width=0.5mm] (2.8, 1.75) -- (2.8, -1.75) ;
 
 \begin{scope} [local bounding box = c1, shift = {(3.85, 0.0)}]
  \draw [ue] (0.0, 1.0) -- (0.0, 1.4) node[el, midway, left] {1} ;
  \draw [ue] (-0.3, 0.7) -- (-0.7, 0.7) node[el, midway, above] {2} -- (-0.7, 0.0) -- (0.0, 0.0) ;
  \draw [ue] (0.3, 0.7) -- (0.7, 0.7) node[el, midway, above] {3} -- (0.7, 0.0) -- (0.0, 0.0) ;
  \draw [ue] (0.0, -0.4) -- (0.0, 0.0) node[el, midway, right] {1} ;
  \draw [hl] (-0.3, 1.0) rectangle (0.3, 0.4) node[midway] {$B$} node[below] {$\alpha$} ;
  \draw [hl] (-0.3, -0.4) rectangle (0.3, -1.0) node[midway] {$A$} node[right] {$\beta$} ;
  \path (0.0, 1.4) pic{ev} node[right] {1} ;
  \path (0.0, 0.0) pic{iv} ;
  \node [pn] at (0.0, -1.4) {$P2$} ;
 \end{scope}
 
 \draw[black, line width=0.5mm] (4.9, 1.75) -- (4.9, -1.75) ;
 
 \begin{scope} [local bounding box = c1, shift = {(5.6, 0.0)}]
  \draw [ue] (0.0, 0.3) -- (0.0, 0.7) node[el, midway, left] {1} ;
  \draw [ht] (-0.3, 0.3) rectangle (0.3, -0.3) node[midway] {$1$} ;
  \path (0.0, 0.7) pic{ev} node[right] {1} ;
  \node [pn] at (0.0, -1.4) {$P3$} ;
 \end{scope}
 
 \node[text width = 1.0cm, align=center] at (7.525, 0.0) {$B::=$} ;
 
 \begin{scope} [local bounding box = c1, shift = {(9.8, 0.0)}]
  \draw [ue] (0.0, 1.0) -- (0.0, 1.4) node[el, midway, left] {1} ;
  \draw [ue] (-0.3, 0.7) -- (-0.7, 0.7) node[el, midway, above] {2} ;
  \draw [ue] (0.3, 0.7) -- (0.7, 0.7) node[el, midway, above] {3} -- (0.7, 0.0) ;
  \draw [ue] (-0.7, 0.3) -- (-0.7, 0.7) node[el, midway, left] {1} ;
  \draw [ue] (-1.0, 0.0) -- (-1.4, 0.0) node[el, midway, above] {2} ;
  \draw [ue] (-0.4, 0.0) -- (0.7, 0.0) node[el, midway, below] {3} ;
  \draw [hl] (-0.3, 1.0) rectangle (0.3, 0.4) node[midway] {$D$} node[below] {$\alpha$} ;
  \draw [hl] (-1.0, 0.3) rectangle (-0.4, -0.3) node[midway] {$B$} node[below] {$\beta$} ;
  \path (0.0, 1.4) pic{ev} node[right] {1} ;
  \path (-1.4, 0.0) pic{ev} node[below] {2} ;
  \path (0.7, 0.0) pic{ev} node[below] {3} ;
  \path (-0.7, 0.7) pic{iv} ;
  \node [pn] at (-0.35, -1.4) {$P4$} ;
 \end{scope}
 
 \draw[black, line width=0.5mm] (10.85, 1.75) -- (10.85, -1.75) ;
 
 \begin{scope} [local bounding box = c1, shift = {(11.9, 0.0)}]
  \draw [ue] (0.0, 1.0) -- (0.0, 1.4) node[el, midway, left] {1} ;
  \draw [ue] (-0.3, 0.7) -- (-0.7, 0.7) node[el, midway, above] {2} -- (-0.7, 0.0) ;
  \draw [ue] (0.3, 0.7) -- (0.7, 0.7) node[el, midway, above] {3} ;
  \draw [ue] (0.7, 0.3) -- (0.7, 0.7) node[el, midway, right] {1} ;
  \draw [ue] (0.4, 0.0) -- (-0.7, 0.0) node[el, midway, below] {2} ;
  \draw [ue] (1.0, 0.0) -- (1.4, 0.0) node[el, midway, above] {3} ;
  \draw [hl] (0.3, 1.0) rectangle (-0.3, 0.4) node[midway] {$D$} node[below] {$\alpha$} ;
  \draw [hl] (0.4, 0.3) rectangle (1.0, -0.3) node[midway] {$B$} node[below] {$\beta$} ;
  \path (0.0, 1.4) pic{ev} node[right] {1} ;
  \path (-0.7, 0.0) pic{ev} node[below] {2} ;
  \path (1.4, 0.0) pic{ev} node[below] {3} ;
  \path (0.7, 0.7) pic{iv} ;
  \node [pn] at (0.5, -1.4) {$P5$} ;
 \end{scope}
 
 \draw[black, line width=0.5mm] (13.65, 1.75) -- (13.65, -1.75) ;
 
 \begin{scope} [local bounding box = c1, shift = {(14.7, 0.0)}]
  \draw [ue] (0.0, 0.3) -- (0.0, 0.7) node[el, midway, left] {1} ;
  \draw [ue] (-0.3, 0.0) -- (-0.7, 0.0) node[el, midway, above] {2} ;
  \draw [ue] (0.3, 0.0) -- (0.7, 0.0) node[el, midway, above] {3} ;
  \draw [ht] (-0.3, 0.3) rectangle (0.3, -0.3) node[midway] {$+$} ;
  \path (0.0, 0.7) pic{ev} node[right] {1} ;
  \path (-0.7, 0.0) pic{ev} node[below] {2} ;
  \path (0.7, 0.0) pic{ev} node[below] {3} ;
  \node [pn] at (0.0, -1.4) {$P6$} ;
 \end{scope}
 
 \draw[black, line width=0.5mm] (15.75, 1.75) -- (15.75, -1.75) ;
 
 \begin{scope} [local bounding box = c1, shift = {(16.8, 0.0)}]
  \draw [ue] (0.0, 0.3) -- (0.0, 0.7) node[el, midway, left] {1} ;
  \draw [ue] (-0.3, 0.0) -- (-0.7, 0.0) node[el, midway, above] {2} ;
  \draw [ue] (0.3, 0.0) -- (0.7, 0.0) node[el, midway, above] {3} ;
  \draw [ht] (-0.3, 0.3) rectangle (0.3, -0.3) node[midway] {$*$} ;
  \path (0.0, 0.7) pic{ev} node[right] {1} ;
  \path (-0.7, 0.0) pic{ev} node[below] {2} ;
  \path (0.7, 0.0) pic{ev} node[below] {3} ;
  \node [pn] at (0.0, -1.4) {$P7$} ;
 \end{scope}
 
 \node[text width = 1.0cm, align=center] at (-0.175, -3.15) {$C::=$} ;
 
 \begin{scope} [local bounding box = c1, shift = {(1.4, -3.15)}]
  \draw [ue] (0.0, 0.3) -- (0.0, 0.7) node[el, midway, left] {1} ;
  \draw [ue] (-0.3, 0.0) -- (-0.7, 0.0) node[el, midway, above] {2} ;
  \draw [ue] (0.3, 0.0) -- (0.7, 0.0) node[el, midway, above] {3} ;
  \draw [ue] (0.7, -0.4) -- (0.7, 0.0) node[el, midway, left] {1} ;
  \draw [hl] (0.3, 0.3) rectangle (-0.3, -0.3) node[midway] {$B$} node[below] {$\alpha$}  ;
  \draw [hl] (0.4, -0.4) rectangle (1.0, -1.0) node[midway] {$A$} node[below] {$\beta$}  ;
  \path (0.0, 0.7) pic{ev} node[right] {1} ;
  \path (-0.7, 0.0) pic{ev} node[below] {2} ;
  \path (0.7, 0.0) pic{ev} node[right] {3} ;
  \node [pn] at (0.0, -1.4) {$P8$} ;
 \end{scope}
 
 \node[text width = 1.0cm, align=center] at (4.025, -3.15) {$D::=$} ;
 
 \begin{scope} [local bounding box = c1, shift = {(5.6, -3.15)}]
  \draw [ue] (0.0, 0.3) -- (0.0, 0.7) node[el, midway, left] {1} ;
  \draw [ue] (-0.3, 0.0) -- (-0.7, 0.0) node[el, midway, above] {2} ;
  \draw [ue] (0.3, 0.0) -- (0.7, 0.0) node[el, midway, above] {3} ;
  \draw [ht] (-0.3, 0.3) rectangle (0.3, -0.3) node[midway] {$+$} ;
  \path (0.0, 0.7) pic{ev} node[right] {1} ;
  \path (-0.7, 0.0) pic{ev} node[below] {2} ;
  \path (0.7, 0.0) pic{ev} node[below] {3} ;
  \node [pn] at (0.0, -1.4) {$P9$} ;
 \end{scope}
 
 \draw[black, line width=0.5mm] (6.65, -2.1) -- (6.65, -4.9) ;
 
 \begin{scope} [local bounding box = c1, shift = {(7.7, -3.15)}]
  \draw [ue] (0.0, 0.3) -- (0.0, 0.7) node[el, midway, left] {1} ;
  \draw [ue] (-0.3, 0.0) -- (-0.7, 0.0) node[el, midway, above] {2} ;
  \draw [ue] (0.3, 0.0) -- (0.7, 0.0) node[el, midway, above] {3} ;
  \draw [ht] (-0.3, 0.3) rectangle (0.3, -0.3) node[midway] {$*$} ;
  \path (0.0, 0.7) pic{ev} node[right] {1} ;
  \path (-0.7, 0.0) pic{ev} node[below] {2} ;
  \path (0.7, 0.0) pic{ev} node[below] {3} ;
  \node [pn] at (0.0, -1.4) {$P10$} ;
 \end{scope}
\end{tikzpicture}%
}
\caption{$\textit{CNF}$ of the grammar in Figure \ref{fig:tggrammar}}
\label{fig:cnftggrammar}
\end{figure}
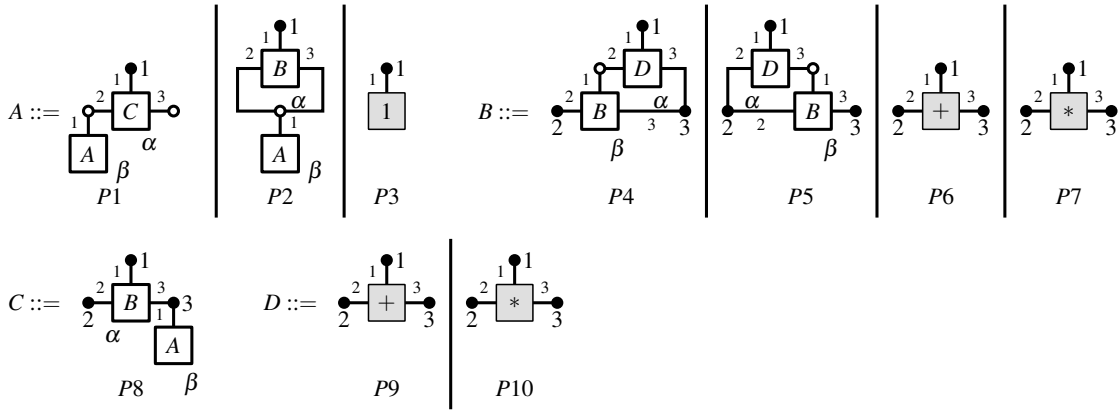

\begin{lemma}\label{lem:cnf}
There exists an algorithm that for every hyperedge replacement grammar $G$ produces a grammar $G'$ in $\textit{CNF}$ such that $L(G) = L(G')$.
\end{lemma}

\begin{proof}
We present a set of rules to transform any grammar $G$, into an equivalent grammar $G'$ such that, for each direct derivation $H \Rightarrow_{p} H'$ with $p \in P_G$, it exists an equivalent derivation $H \Rightarrow^*_Q H'$ with $Q \subseteq P_{G'}$. The proof is provided with a running example showing the application of the rules. The grammar in Figure \ref{fig:cnf000} contains productions that are not in $\textit{CNF}$: $P1$ has more than $2$ hyperedges; $P2$ has a single non-terminal hyperedge; $P3$ is an empty production, but its \textit{lhs} is not $S$; $P4$ has 2 hyperedges one of which is terminal.

\begin{figure}[htpb]
\centering
\resizebox{10cm}{!}{%
\begin{tikzpicture}
 \node[text width = 1.0cm, align=center] at (-0.175, 0.0) {$S::=$} ;
 
 \begin{scope} [local bounding box = c1, shift = {(1.4, 0.0)}]
  \draw [ue] (-0.7, -0.3) -- (-0.7, -0.7) node[el, midway, left] {1} ;
  \draw [ue] (-0.4, 0.0) -- (0.0, 0.0) node[el, midway, above] {2} ;
  \draw [ue] (-0.7, 0.3) -- (-0.7, 0.7) node[el, midway, left] {3} ;
  \draw [ue] (0.0, -0.4) -- (0.0, 0.0) node[el, midway, right] {1} ;
  \draw [ue] (0.3, -0.7) -- (0.7, -0.7) node[el, midway, below] {2} ;
  \draw [ue] (-0.3, -0.7) -- (-0.7, -0.7) node[el, midway, below] {3} ;
  \draw [ue] (0.0, 0.4) -- (0.0, 0.0) node[el, midway, right] {1} ;
  \draw [ue] (-0.3, 0.7) -- (-0.7, 0.7) node[el, midway, above] {2} ;
  \draw [hl] (-1.0, -0.3) rectangle (-0.4, 0.3) node[midway] {$B$} ;
  \draw [hl] (-0.3, -1.0) rectangle (0.3, -0.4) node[midway] {$B$} ;
  \draw [hl] (-0.3, 0.4) rectangle (0.3, 1.0) node[midway] {$C$} ;
  \path (-0.7, 0.7) pic{ev} node[above left] {1} ;
  \path (-0.7, -0.7) pic{ev} node[below left] {2} ;
  \path (0.0, 0.0) pic{iv} ;
  \path (0.7, -0.7) pic{iv} ;
  \node [pn] at (0.0, -1.4) {$P1$} ;
 \end{scope}
 
 \node[text width = 1.0cm, align=center] at (3.325, 0.0) {$B::=$} ;
 
 \begin{scope} [local bounding box = c1, shift = {(4.9, 0.0)}]
  \draw [ue] (-0.7, 0.3) -- (-0.7, 0.7) node[el, midway, left] {1} ;
  \draw [ue] (-0.7, -0.3) -- (-0.7, -0.7) node[el, midway, left] {2} ;
  \draw [hl] (-1.0, -0.3) rectangle (-0.4, 0.3) node[midway] {$C$} ;
  \path (-0.7, 0.7) pic{ev} node[right] {1} ;
  \path (0.0, 0.0) pic{ev} node[above] {2} ;
  \path (-0.7, -0.7) pic{ev} node[right] {3} ;
  \node [pn] at (-0.35, -1.4) {$P2$} ;
 \end{scope}
 
 \draw[black, line width=0.5mm] (5.25, 1.05) -- (5.25, -1.75) ;
 
 \begin{scope} [local bounding box = c1, shift = {(5.6, 0.0)}]
  \path (0.0, 0.7) pic{ev} node[right] {1} ;
  \path (0.0, 0.0) pic{ev} node[right] {2} ;
  \path (0.0, -0.7) pic{ev} node[right] {3} ;
  \node [pn] at (0.0, -1.4) {$P3$} ;
 \end{scope}
 
 \node[text width = 1.0cm, align=center] at (6.825, 0.0) {$C::=$} ;
 
 \begin{scope} [local bounding box = c1, shift = {(7.7, 0.0)}]
  \draw [ue] (0.0, 0.3) -- (0.0, 0.7) node[el, midway, left] {1} ;
  \draw [ue] (0.0, -0.3) -- (0.0, -0.7) node[el, midway, left] {2} ;
  \draw [ue] (0.4, -0.7) -- (0.0, -0.7) node[el, midway, below] {1} ;
  \draw [ue] (0.7, -0.4) -- (0.7, 0.0) node[el, midway, left] {2} ;
  \draw [ht] (-0.3, -0.3) rectangle (0.3, 0.3) node[midway] {$a$} ;
  \draw [hl] (0.4, -1.0) rectangle (1.0, -0.4) node[midway] {$S$} ;
  \path (0.0, 0.7) pic{ev} node[right] {1} ;
  \path (0.7, 0.0) pic{ev} node[right] {2} ;
  \path (0.0, -0.7) pic{iv} ;
  \node [pn] at (0.5, -1.4) {$P4$} ;
 \end{scope}
 
 \draw[black, line width=0.5mm] (9.1, 1.05) -- (9.1, -1.75) ;
 
 \begin{scope} [local bounding box = c1, shift = {(9.8, 0.0)}]
  \draw [ue] (0.0, 0.3) -- (0.0, 0.7) node[el, midway, left] {1} ;
  \draw [ue] (0.0, -0.3) -- (0.0, -0.7) node[el, midway, left] {2} ;
  \draw [ht] (-0.3, -0.3) rectangle (0.3, 0.3) node[midway] {$c$} ;
  \path (0.0, 0.7) pic{ev} node[right] {1} ;
  \path (0.0, -0.7) pic{ev} node[right] {2} ;
  \node [pn] at (0.0, -1.4) {$P5$} ;
 \end{scope}
\end{tikzpicture}%
}
\caption{Starting grammar for the proof of \textit{CNF} equivalence.}
\label{fig:cnf000}
\end{figure}
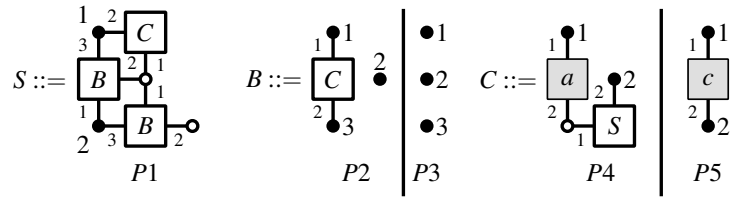

For a production $p = (A, R) \in P$, that is not already in \textit{CNF}, we consider the following set of rules, applied in this order, to obtain a corresponding equivalent set of productions $P'$ in \textit{CNF}:

\begin{enumerate}
\item[1.] If $p$ is the empty production, for each production $q = (B, X) \in P$ having $e \in E_{X}$ with $\textit{lab}(e) = A$ in its \textit{rhs}, for each production $q' = (A, Y) \in P$ having $A$ in its \textit{lhs} we apply the substitution $R' = X[e, Y]$ and add the productions $p = (B, R')$. We then remove the productions that are no longer needed. The proof of equivalence of the derivations $H \Rightarrow_{q} H' \Rightarrow_{q'} H''$ and $H \Rightarrow_{p'} H''$ is the following: if $e'$ with $\textit{lab}(e') = B$ is the hyperedge involved in the derivation $H \Rightarrow_{q} H'$ then $H'' = H[e'/X[e/Y]] = H[e'/R']$ since $R' = X[e, Y]$.

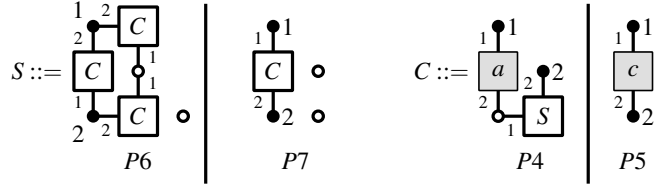
\begin{figure}[htpb]
\centering
\resizebox{9cm}{!}{%
\begin{tikzpicture}
 \node[text width = 1.0cm, align=center] at (-0.175, 0.0) {$S::=$} ;
 
 \begin{scope} [local bounding box = c1, shift = {(1.4, 0.0)}]
  \draw [ue] (-0.7, -0.3) -- (-0.7, -0.7) node[el, midway, left] {1} ;
  \draw [ue] (-0.7, 0.3) -- (-0.7, 0.7) node[el, midway, left] {2} ;
  \draw [ue] (0.0, -0.4) -- (0.0, 0.0) node[el, midway, right] {1} ;
  \draw [ue] (-0.3, -0.7) -- (-0.7, -0.7) node[el, midway, below] {2} ;
  \draw [ue] (0.0, 0.4) -- (0.0, 0.0) node[el, midway, right] {1} ;
  \draw [ue] (-0.3, 0.7) -- (-0.7, 0.7) node[el, midway, above] {2} ;
  \draw [hl] (-1.0, -0.3) rectangle (-0.4, 0.3) node[midway] {$C$} ;
  \draw [hl] (-0.3, -1.0) rectangle (0.3, -0.4) node[midway] {$C$} ;
  \draw [hl] (-0.3, 0.4) rectangle (0.3, 1.0) node[midway] {$C$} ;
  \path (-0.7, 0.7) pic{ev} node[above left] {1} ;
  \path (-0.7, -0.7) pic{ev} node[below left] {2} ;
  \path (0.0, 0.0) pic{iv} ;
  \path (0.7, -0.7) pic{iv} ;
  \node [pn] at (0.0, -1.4) {$P6$} ;
 \end{scope}
 
 \draw[black, line width=0.5mm] (2.45, 1.05) -- (2.45, -1.75) ;
 
 \begin{scope} [local bounding box = c1, shift = {(4.2, 0.0)}]
  \draw [ue] (-0.7, 0.3) -- (-0.7, 0.7) node[el, midway, left] {1} ;
  \draw [ue] (-0.7, -0.3) -- (-0.7, -0.7) node[el, midway, left] {2} ;
  \draw [hl] (-1.0, -0.3) rectangle (-0.4, 0.3) node[midway] {$C$} ;
  \path (-0.7, 0.7) pic{ev} node[right] {1} ;
  \path (-0.7, -0.7) pic{ev} node[right] {2} ;
  \path (0.0, 0.0) pic{iv} ;
  \path (0.0, -0.7) pic{iv} ;
  \node [pn] at (-0.35, -1.4) {$P7$} ;
 \end{scope}
 
 \node[text width = 1.0cm, align=center] at (6.125, 0.0) {$C::=$} ;
 
 \begin{scope} [local bounding box = c1, shift = {(7.0, 0.0)}]
  \draw [ue] (0.0, 0.3) -- (0.0, 0.7) node[el, midway, left] {1} ;
  \draw [ue] (0.0, -0.3) -- (0.0, -0.7) node[el, midway, left] {2} ;
  \draw [ue] (0.4, -0.7) -- (0.0, -0.7) node[el, midway, below] {1} ;
  \draw [ue] (0.7, -0.4) -- (0.7, 0.0) node[el, midway, left] {2} ;
  \draw [ht] (-0.3, -0.3) rectangle (0.3, 0.3) node[midway] {$a$} ;
  \draw [hl] (0.4, -1.0) rectangle (1.0, -0.4) node[midway] {$S$} ;
  \path (0.0, 0.7) pic{ev} node[right] {1} ;
  \path (0.7, 0.0) pic{ev} node[right] {2} ;
  \path (0.0, -0.7) pic{iv} ;
  \node [pn] at (0.5, -1.4) {$P4$} ;
 \end{scope}
 
 \draw[black, line width=0.5mm] (8.4, 1.05) -- (8.4, -1.75) ;
 
 \begin{scope} [local bounding box = c1, shift = {(9.1, 0.0)}]
  \draw [ue] (0.0, 0.3) -- (0.0, 0.7) node[el, midway, left] {1} ;
  \draw [ue] (0.0, -0.3) -- (0.0, -0.7) node[el, midway, left] {2} ;
  \draw [ht] (-0.3, -0.3) rectangle (0.3, 0.3) node[midway] {$c$} ;
  \path (0.0, 0.7) pic{ev} node[right] {1} ;
  \path (0.0, -0.7) pic{ev} node[right] {2} ;
  \node [pn] at (0.0, -1.4) {$P5$} ;
 \end{scope}
\end{tikzpicture}%
}
\caption{Removal of the empty production $P3$.}
\label{fig:cnf001}
\end{figure}

In order to remove the empty production $P3$ (Fig. \ref{fig:cnf001}) we apply the replacements of all the productions having $B$ as their $\textit{lhs}$ to all the productions having a hyperedge labelled as $B$ in their $\textit{rhs}$. We remove $P1$ and introduce the productions $P6$ and $P7$. We then remove $P2$ and $P3$ since they are no longer needed.

\item[2.] If $E_R = \{e'\}$ with $\textit{lab}(e') \in N$ for each production $q = (\textit{lab}(e'), X) \in P$ we add the production $p' = (\textit{lab}(e), R')$ with $R' = R[e'/X]$. If $E_{R'} = \{e''\}$ with $\textit{lab}(e'') \in N$ this step is iterated and terminates when $|E_{R'}| > 1$ or $E_{R'} = \{e_t\}$ with $\textit{lab}(e_t) \in \Sigma$ or $|E_{R'}| = 0$ and $|V_{R'}| > \textit{ext}_{R'}$. The proof of equivalence of the derivations $H \Rightarrow_{p} H' \Rightarrow_{q} H''$ and $H \Rightarrow_{p'} H''$ is the following: if $e'$ is the hyperedge involved in the derivation $H' \Rightarrow_{q} H''$ then $H'' = H'[e/R[e'/X]] = H[e/R']$ since $R' = R[e'/X]$.

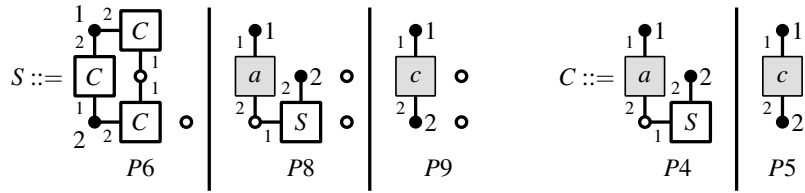
\begin{figure}[htpb]
\centering
\resizebox{11cm}{!}{%
\begin{tikzpicture}
 \node[text width = 1.0cm, align=center] at (-0.175, 0.0) {$S::=$} ;
 
 \begin{scope} [local bounding box = c1, shift = {(1.4, 0.0)}]
  \draw [ue] (-0.7, -0.3) -- (-0.7, -0.7) node[el, midway, left] {1} ;
  \draw [ue] (-0.7, 0.3) -- (-0.7, 0.7) node[el, midway, left] {2} ;
  \draw [ue] (0.0, -0.4) -- (0.0, 0.0) node[el, midway, right] {1} ;
  \draw [ue] (-0.3, -0.7) -- (-0.7, -0.7) node[el, midway, below] {2} ;
  \draw [ue] (0.0, 0.4) -- (0.0, 0.0) node[el, midway, right] {1} ;
  \draw [ue] (-0.3, 0.7) -- (-0.7, 0.7) node[el, midway, above] {2} ;
  \draw [hl] (-1.0, -0.3) rectangle (-0.4, 0.3) node[midway] {$C$} ;
  \draw [hl] (-0.3, -1.0) rectangle (0.3, -0.4) node[midway] {$C$} ;
  \draw [hl] (-0.3, 0.4) rectangle (0.3, 1.0) node[midway] {$C$} ;
  \path (-0.7, 0.7) pic{ev} node[above left] {1} ;
  \path (-0.7, -0.7) pic{ev} node[below left] {2} ;
  \path (0.0, 0.0) pic{iv} ;
  \path (0.7, -0.7) pic{iv} ;
  \node [pn] at (0.0, -1.4) {$P6$} ;
 \end{scope}
 
 \draw[black, line width=0.5mm] (2.45, 1.05) -- (2.45, -1.75) ;
 
 \begin{scope} [local bounding box = c1, shift = {(3.15, 0.0)}]
  \draw [ue] (0.0, 0.3) -- (0.0, 0.7) node[el, midway, left] {1} ;
  \draw [ue] (0.0, -0.3) -- (0.0, -0.7) node[el, midway, left] {2} ;
  \draw [ue] (0.4, -0.7) -- (0.0, -0.7) node[el, midway, below] {1} ;
  \draw [ue] (0.7, -0.4) -- (0.7, 0.0) node[el, midway, left] {2} ;
  \draw [ht] (-0.3, -0.3) rectangle (0.3, 0.3) node[midway] {$a$} ;
  \draw [hl] (0.4, -1.0) rectangle (1.0, -0.4) node[midway] {$S$} ;
  \path (0.0, 0.7) pic{ev} node[right] {1} ;
  \path (0.7, 0.0) pic{ev} node[right] {2} ;
  \path (0.0, -0.7) pic{iv} ;
  \path (1.4, 0.0) pic{iv} ;
  \path (1.4, -0.7) pic{iv} ;
  \node [pn] at (0.7, -1.4) {$P8$} ;
 \end{scope}
 
 \draw[black, line width=0.5mm] (4.9, 1.05) -- (4.9, -1.75) ;
 
 \begin{scope} [local bounding box = c1, shift = {(6.3, 0.0)}]
  \draw [ue] (-0.7, 0.3) -- (-0.7, 0.7) node[el, midway, left] {1} ;
  \draw [ue] (-0.7, -0.3) -- (-0.7, -0.7) node[el, midway, left] {2} ;
  \draw [ht] (-1.0, -0.3) rectangle (-0.4, 0.3) node[midway] {$c$} ;
  \path (-0.7, 0.7) pic{ev} node[right] {1} ;
  \path (-0.7, -0.7) pic{ev} node[right] {2} ;
  \path (0.0, 0.0) pic{iv} ;
  \path (0.0, -0.7) pic{iv} ;
  \node [pn] at (-0.35, -1.4) {$P9$} ;
 \end{scope}
 
 \node[text width = 1.0cm, align=center] at (8.225, 0.0) {$C::=$} ;
 
 \begin{scope} [local bounding box = c1, shift = {(9.1, 0.0)}]
  \draw [ue] (0.0, 0.3) -- (0.0, 0.7) node[el, midway, left] {1} ;
  \draw [ue] (0.0, -0.3) -- (0.0, -0.7) node[el, midway, left] {2} ;
  \draw [ue] (0.4, -0.7) -- (0.0, -0.7) node[el, midway, below] {1} ;
  \draw [ue] (0.7, -0.4) -- (0.7, 0.0) node[el, midway, left] {2} ;
  \draw [ht] (-0.3, -0.3) rectangle (0.3, 0.3) node[midway] {$a$} ;
  \draw [hl] (0.4, -1.0) rectangle (1.0, -0.4) node[midway] {$S$} ;
  \path (0.0, 0.7) pic{ev} node[right] {1} ;
  \path (0.7, 0.0) pic{ev} node[right] {2} ;
  \path (0.0, -0.7) pic{iv} ;
  \node [pn] at (0.5, -1.4) {$P4$} ;
 \end{scope}
 
 \draw[black, line width=0.5mm] (10.5, 1.05) -- (10.5, -1.75) ;
 
 \begin{scope} [local bounding box = c1, shift = {(11.2, 0.0)}]
  \draw [ue] (0.0, 0.3) -- (0.0, 0.7) node[el, midway, left] {1} ;
  \draw [ue] (0.0, -0.3) -- (0.0, -0.7) node[el, midway, left] {2} ;
  \draw [ht] (-0.3, -0.3) rectangle (0.3, 0.3) node[midway] {$c$} ;
  \path (0.0, 0.7) pic{ev} node[right] {1} ;
  \path (0.0, -0.7) pic{ev} node[right] {2} ;
  \node [pn] at (0.0, -1.4) {$P5$} ;
 \end{scope}
\end{tikzpicture}%
}
\caption{Removal of production $P7$}
\label{fig:cnf002}
\end{figure}

Since $P7$ has a single non-terminal hyperedge $C$ (Fig. \ref{fig:cnf002}), we apply a replacement for each production that has $C$ on its $\textit{lhs}$. In our case, using the replacements of $P4$ and $P5$, we obtain $P8$ and $P9$. The production $P7$ is removed from the grammar.

\item[3.] If $|E_R| = k > 2$ we consider the subgraph $X$ of $R$ composed by the subset $E_X \subset E_R$ of hyperedges $e_2, \ldots, e_k$ and their attachment nodes. We introduce a new label $T$ so that $N' = N \cup \{T\}$ and a new handle $T^{\bullet}$ of $e_T$ with $\textit{ext}(e_T) = \underset{2 \leq i \leq k}{\bigcup} \textit{att}(e_i)$ such that $\textit{type}(e_T) = \textit{type}(X)$. We then consider the hypergraph $R'$ composed by $R \backslash X$ and $T^{\bullet}$ where $V_{R'} = V_{R \backslash X} \cup V_{T^{\bullet}}$ and $E_{R'} = E_{R \backslash X} \cup E_{T^{\bullet}}$. Finally we add the productions $p' = (A, R'), p'' = (T, X)$ to $P'$. If $|E_X| > 2$ this step is iterated. The proof of equivalence of the derivations $H \Rightarrow_{p} H'$ and $H \Rightarrow^*_{P'} H'$ is the following: if $e_a$ is the handle of the \textit{lhs} of $p$ we consider the following equivalence of the replacements then $H' = H[e_a/R] = H[e_a/R'[e_T/X]]$ since $R = R'[e_T/X]$.

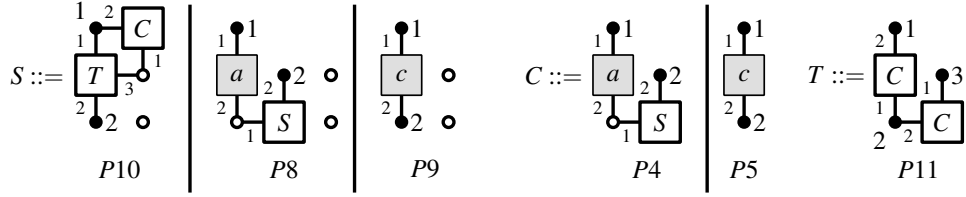
\begin{figure}[htpb]
\centering
\resizebox{13cm}{!}{%
\begin{tikzpicture}
 \node[text width = 1.0cm, align=center] at (-0.175, 0.0) {$S::=$} ;
 
 \begin{scope} [local bounding box = c1, shift = {(1.4, 0.0)}]
  \draw [ue] (-0.7, -0.3) -- (-0.7, -0.7) node[el, midway, left] {2} ;
  \draw [ue] (-0.4, 0.0) -- (0.0, 0.0) node[el, midway, below] {3} ;
  \draw [ue] (-0.7, 0.3) -- (-0.7, 0.7) node[el, midway, left] {1} ;
  \draw [ue] (0.0, 0.4) -- (0.0, 0.0) node[el, midway, right] {1} ;
  \draw [ue] (-0.3, 0.7) -- (-0.7, 0.7) node[el, midway, above] {2} ;
  \draw [hl] (-1.0, -0.3) rectangle (-0.4, 0.3) node[midway] {$T$} ;
  \draw [hl] (-0.3, 0.4) rectangle (0.3, 1.0) node[midway] {$C$} ;
  \path (-0.7, 0.7) pic{ev} node[above left] {1} ;
  \path (-0.7, -0.7) pic{ev} node[right] {2} ;
  \path (0.0, 0.0) pic{iv} ;
  \path (0.0, -0.7) pic{iv} ;
  \node [pn] at (-0.35, -1.4) {$P10$} ;
 \end{scope}
 
 \draw[black, line width=0.5mm] (2.1, 1.05) -- (2.1, -1.75) ;
 
 \begin{scope} [local bounding box = c1, shift = {(2.8, 0.0)}]
  \draw [ue] (0.0, 0.3) -- (0.0, 0.7) node[el, midway, left] {1} ;
  \draw [ue] (0.0, -0.3) -- (0.0, -0.7) node[el, midway, left] {2} ;
  \draw [ue] (0.4, -0.7) -- (0.0, -0.7) node[el, midway, below] {1} ;
  \draw [ue] (0.7, -0.4) -- (0.7, 0.0) node[el, midway, left] {2} ;
  \draw [ht] (-0.3, -0.3) rectangle (0.3, 0.3) node[midway] {$a$} ;
  \draw [hl] (0.4, -1.0) rectangle (1.0, -0.4) node[midway] {$S$} ;
  \path (0.0, 0.7) pic{ev} node[right] {1} ;
  \path (0.7, 0.0) pic{ev} node[right] {2} ;
  \path (0.0, -0.7) pic{iv} ;
  \path (1.4, 0.0) pic{iv} ;
  \path (1.4, -0.7) pic{iv} ;
  \node [pn] at (0.7, -1.4) {$P8$} ;
 \end{scope}
 
 \draw[black, line width=0.5mm] (4.55, 1.05) -- (4.55, -1.75) ;
 
 \begin{scope} [local bounding box = c1, shift = {(5.95, 0.0)}]
  \draw [ue] (-0.7, 0.3) -- (-0.7, 0.7) node[el, midway, left] {1} ;
  \draw [ue] (-0.7, -0.3) -- (-0.7, -0.7) node[el, midway, left] {2} ;
  \draw [ht] (-1.0, -0.3) rectangle (-0.4, 0.3) node[midway] {$c$} ;
  \path (-0.7, 0.7) pic{ev} node[right] {1} ;
  \path (-0.7, -0.7) pic{ev} node[right] {2} ;
  \path (0.0, 0.0) pic{iv} ;
  \path (0.0, -0.7) pic{iv} ;
  \node [pn] at (-0.35, -1.4) {$P9$} ;
 \end{scope}
 
 \node[text width = 1.0cm, align=center] at (7.525, 0.0) {$C::=$} ;
 
 \begin{scope} [local bounding box = c1, shift = {(8.4, 0.0)}]
  \draw [ue] (0.0, 0.3) -- (0.0, 0.7) node[el, midway, left] {1} ;
  \draw [ue] (0.0, -0.3) -- (0.0, -0.7) node[el, midway, left] {2} ;
  \draw [ue] (0.4, -0.7) -- (0.0, -0.7) node[el, midway, below] {1} ;
  \draw [ue] (0.7, -0.4) -- (0.7, 0.0) node[el, midway, left] {2} ;
  \draw [ht] (-0.3, -0.3) rectangle (0.3, 0.3) node[midway] {$a$} ;
  \draw [hl] (0.4, -1.0) rectangle (1.0, -0.4) node[midway] {$S$} ;
  \path (0.0, 0.7) pic{ev} node[right] {1} ;
  \path (0.7, 0.0) pic{ev} node[right] {2} ;
  \path (0.0, -0.7) pic{iv} ;
  \node [pn] at (0.5, -1.4) {$P4$} ;
 \end{scope}
 
 \draw[black, line width=0.5mm] (9.8, 1.05) -- (9.8, -1.75) ;
 
 \begin{scope} [local bounding box = c1, shift = {(10.35, 0.0)}]
  \draw [ue] (0.0, 0.3) -- (0.0, 0.7) node[el, midway, left] {1} ;
  \draw [ue] (0.0, -0.3) -- (0.0, -0.7) node[el, midway, left] {2} ;
  \draw [ht] (-0.3, -0.3) rectangle (0.3, 0.3) node[midway] {$c$} ;
  \path (0.0, 0.7) pic{ev} node[right] {1} ;
  \path (0.0, -0.7) pic{ev} node[right] {2} ;
  \node [pn] at (0.0, -1.4) {$P5$} ;
 \end{scope}
 
 \node[text width = 1.0cm, align=center] at (11.725, -0.0) {$T::=$} ;
 
 \begin{scope} [local bounding box = c1, shift = {(13.3, 0.0)}]
  \draw [ue] (-0.7, -0.3) -- (-0.7, -0.7) node[el, midway, left] {1} ;
  \draw [ue] (-0.7, 0.3) -- (-0.7, 0.7) node[el, midway, left] {2} ;
  \draw [ue] (0.0, -0.4) -- (0.0, 0.0) node[el, midway, left] {1} ;
  \draw [ue] (-0.3, -0.7) -- (-0.7, -0.7) node[el, midway, below] {2} ;
  \draw [hl] (-1.0, -0.3) rectangle (-0.4, 0.3) node[midway] {$C$} ;
  \draw [hl] (-0.3, -0.4) rectangle (0.3, -1.0) node[midway] {$C$} ;
  \path (-0.7, 0.7) pic{ev} node[right] {1} ;
  \path (-0.7, -0.7) pic{ev} node[below left] {2} ;
  \path (0.0, 0.0) pic{ev} node[right] {3} ;
  \node [pn] at (-0.35, -1.4) {$P11$} ;
 \end{scope}
\end{tikzpicture}%
}
\caption{Removal of production $P6$}
\label{fig:cnf003}
\end{figure}

Since production $P6$ has three non-terminal hyperedges (Fig. \ref{fig:cnf003}), we create a new label $T$, a new handle $T^{\bullet}$ and the production $P11$. Then we add the production $P10$ so that the replacement of the hyperedge labelled as $T$ by the \textit{rhs} of $P11$ results in the \textit{rhs} of $P6$. The production $P6$ is then removed from the grammar.

\item[4.] If $|E_R| > 1$ and exists $e' \in E_R$ such that $\textit{lab}(e') \in \Sigma$ a new label $T$ is introduced so that $N' = N \cup \{T\}$. We add 2 new productions $p' = (A, R')$ to $P'$ where $R' = R$ with $\textit{lab}(e') = T$ and $p'' = (T, e^{'\bullet})$. This step is repeated for each $e' \in E_R$ with $\textit{lab}(e') \in \Sigma$. Due to the confluence property \cite{courcelle-1987-aad} of $\textit{HRG}$s the order in which the terminal hyperedges are chosen is irrelevant. The proof of equivalence of the derivations $H \Rightarrow_{p} H''$ and $H \Rightarrow_{p'} H' \Rightarrow_{p''} H''$ is the following: if $e' \in E_R$ with $\textit{lab}(e') \in \Sigma$ is the hyperedge involved in the derivation $H \Rightarrow_{p} H''$ then $H' = H[e/R] = H[e/R'[e'/e^{'\bullet}]]$ since $R = R'[e'/e^{'\bullet}]$.

\begin{figure}[htpb]
\centering
\resizebox{14cm}{!}{%
\begin{tikzpicture}
 \node[text width = 1.0cm, align=center] at (-0.175, 0.0) {$S::=$} ;
 
 \begin{scope} [local bounding box = c1, shift = {(1.4, 0.0)}]
  \draw [ue] (-0.7, -0.3) -- (-0.7, -0.7) node[el, midway, left] {2} ;
  \draw [ue] (-0.4, 0.0) -- (0.0, 0.0) node[el, midway, below] {3} ;
  \draw [ue] (-0.7, 0.3) -- (-0.7, 0.7) node[el, midway, left] {1} ;
  \draw [ue] (0.0, 0.4) -- (0.0, 0.0) node[el, midway, right] {1} ;
  \draw [ue] (-0.3, 0.7) -- (-0.7, 0.7) node[el, midway, above] {2} ;
  \draw [hl] (-1.0, -0.3) rectangle (-0.4, 0.3) node[midway] {$T$} ;
  \draw [hl] (-0.3, 0.4) rectangle (0.3, 1.0) node[midway] {$C$} ;
  \path (-0.7, 0.7) pic{ev} node[above left] {1} ;
  \path (-0.7, -0.7) pic{ev} node[right] {2} ;
  \path (0.0, 0.0) pic{iv} ;
  \path (0.0, -0.7) pic{iv} ;
  \node [pn] at (-0.35, -1.4) {$P10$} ;
 \end{scope}
 
 \draw[black, line width=0.5mm] (2.1, 1.05) -- (2.1, -1.75) ;
 
 \begin{scope} [local bounding box = c1, shift = {(2.8, 0.0)}]
  \draw [ue] (0.0, 0.3) -- (0.0, 0.7) node[el, midway, left] {1} ;
  \draw [ue] (0.0, -0.3) -- (0.0, -0.7) node[el, midway, left] {2} ;
  \draw [ue] (0.4, -0.7) -- (0.0, -0.7) node[el, midway, below] {1} ;
  \draw [ue] (0.7, -0.4) -- (0.7, 0.0) node[el, midway, left] {2} ;
  \draw [hl] (-0.3, -0.3) rectangle (0.3, 0.3) node[midway] {$A$} ;
  \draw [hl] (0.4, -1.0) rectangle (1.0, -0.4) node[midway] {$S$} ;
  \path (0.0, 0.7) pic{ev} node[right] {1} ;
  \path (0.7, 0.0) pic{ev} node[right] {2} ;
  \path (0.0, -0.7) pic{iv} ;
  \path (1.4, 0.0) pic{iv} ;
  \path (1.4, -0.7) pic{iv} ;
  \node [pn] at (0.7, -1.4) {$P8$} ;
 \end{scope}
 
 \draw[black, line width=0.5mm] (4.55, 1.05) -- (4.55, -1.75) ;
 
 \begin{scope} [local bounding box = c1, shift = {(5.95, 0.0)}]
  \draw [ue] (-0.7, 0.3) -- (-0.7, 0.7) node[el, midway, left] {1} ;
  \draw [ue] (-0.7, -0.3) -- (-0.7, -0.7) node[el, midway, left] {2} ;
  \draw [ht] (-1.0, -0.3) rectangle (-0.4, 0.3) node[midway] {$c$} ;
  \path (-0.7, 0.7) pic{ev} node[right] {1} ;
  \path (-0.7, -0.7) pic{ev} node[right] {2} ;
  \path (0.0, 0.0) pic{iv} ;
  \path (0.0, -0.7) pic{iv} ;
  \node [pn] at (-0.35, -1.4) {$P9$} ;
 \end{scope}
 
 \node[text width = 1.0cm, align=center] at (7.525, 0.0) {$C::=$} ;
 
 \begin{scope} [local bounding box = c1, shift = {(8.4, 0.0)}]
  \draw [ue] (0.0, 0.3) -- (0.0, 0.7) node[el, midway, left] {1} ;
  \draw [ue] (0.0, -0.3) -- (0.0, -0.7) node[el, midway, left] {2} ;
  \draw [ue] (0.4, -0.7) -- (0.0, -0.7) node[el, midway, below] {1} ;
  \draw [ue] (0.7, -0.4) -- (0.7, 0.0) node[el, midway, left] {2} ;
  \draw [hl] (-0.3, -0.3) rectangle (0.3, 0.3) node[midway] {$A$} ;
  \draw [hl] (0.4, -1.0) rectangle (1.0, -0.4) node[midway] {$S$} ;
  \path (0.0, 0.7) pic{ev} node[right] {1} ;
  \path (0.7, 0.0) pic{ev} node[right] {2} ;
  \path (0.0, -0.7) pic{iv} ;
  \node [pn] at (0.5, -1.4) {$P4$} ;
 \end{scope}
 
 \draw[black, line width=0.5mm] (9.8, 1.05) -- (9.8, -1.75) ;
 
 \begin{scope} [local bounding box = c1, shift = {(10.35, 0.0)}]
  \draw [ue] (0.0, 0.3) -- (0.0, 0.7) node[el, midway, left] {1} ;
  \draw [ue] (0.0, -0.3) -- (0.0, -0.7) node[el, midway, left] {2} ;
  \draw [ht] (-0.3, -0.3) rectangle (0.3, 0.3) node[midway] {$c$} ;
  \path (0.0, 0.7) pic{ev} node[right] {1} ;
  \path (0.0, -0.7) pic{ev} node[right] {2} ;
  \node [pn] at (0.0, -1.4) {$P5$} ;
 \end{scope}
 
 \node[text width = 1.0cm, align=center] at (11.725, -0.0) {$T::=$} ;
 
 \begin{scope} [local bounding box = c1, shift = {(13.3, 0.0)}]
  \draw [ue] (-0.7, -0.3) -- (-0.7, -0.7) node[el, midway, left] {1} ;
  \draw [ue] (-0.7, 0.3) -- (-0.7, 0.7) node[el, midway, left] {2} ;
  \draw [ue] (0.0, -0.4) -- (0.0, 0.0) node[el, midway, left] {1} ;
  \draw [ue] (-0.3, -0.7) -- (-0.7, -0.7) node[el, midway, below] {2} ;
  \draw [hl] (-1.0, -0.3) rectangle (-0.4, 0.3) node[midway] {$C$} ;
  \draw [hl] (-0.3, -0.4) rectangle (0.3, -1.0) node[midway] {$C$} ;
  \path (-0.7, 0.7) pic{ev} node[right] {1} ;
  \path (-0.7, -0.7) pic{ev} node[below left] {2} ;
  \path (0.0, 0.0) pic{ev} node[right] {3} ;
  \node [pn] at (-0.35, -1.4) {$P11$} ;
 \end{scope}
 
 \node[text width = 1.0cm, align=center] at (14.525, 0.0) {$A::=$} ;
 
 \begin{scope} [local bounding box = c1, shift = {(15.4, 0.0)}]
  \draw [ue] (0.0, 0.3) -- (0.0, 0.7) node[el, midway, left] {1} ;
  \draw [ue] (0.0, -0.3) -- (0.0, -0.7) node[el, midway, left] {2} ;
  \draw [ht] (-0.3, -0.3) rectangle (0.3, 0.3) node[midway] {$a$} ;
  \path (0.0, 0.7) pic{ev} node[right] {1} ;
  \path (0.0, -0.7) pic{ev} node[right] {2} ;
  \node [pn] at (0.0, -1.4) {$P14$} ;
 \end{scope}
\end{tikzpicture}%
}
\caption{Removal of productions $P8$ and $P4$}
\label{fig:cnf004}
\end{figure}
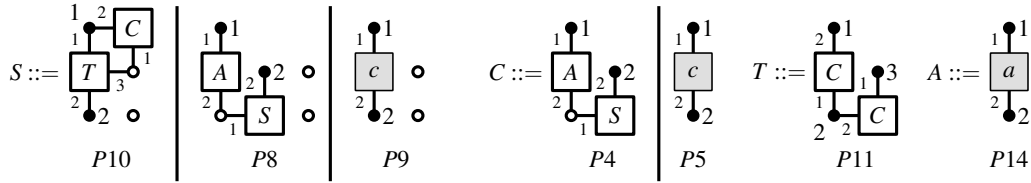

Both \textit{rhs} of productions $P4$ and $P8$ are composed by a terminal and a non-terminal hyperedge. We introduce a new label $A$ and a its handle $A^{\bullet}$ along with the production $P14$ (Figure \ref{fig:cnf004}). We then add the productions $P12$ and $P13$ resulting from the substitution of the terminal hyperedges labelled with $a$ by the non-terminal hyperedges labelled with $A$. Productions $P4$ and $P8$ are then removed from the grammar.

\end{enumerate}

\end{proof}

From this point on, if not explicitly specified, we always refer to an \textit{HRG} as an \textit{HRG} in \textit{CNF}. We stress that the input of the method must be already provided in this form, that is, the time required for the transformation is not taken into account during the evaluation of the time complexity.

In order to complete the adaptation of the grammar we propose a more suitable short-hand representation of the productions that only extracts the necessary information, so, for each $p = (A,R) \in P$ and $i \in \mathbb{N}_0$ we use the following notations:

\begin{itemize}
\setlength\itemsep{-0.2em}
\item{} $A \overset{p}{\longrightarrow} BC, i$ for a \textit{non-terminal} production where $B, C \in N$ are the labels of the marked hyperedges $e_\alpha, e_\beta \in R$ with $\textit{mark}(e_\alpha) = \alpha$, $\textit{mark}(e_\beta) = \beta$ and $i = |V_R \backslash \textit{ext}_R|$.

\item{} $A \overset{p}{\longrightarrow} a, i$ for a \textit{terminal} production where $a \in \Sigma$ is the label of the marked hyperedge $e_{\alpha} \in R$ and $i = |V_R\backslash \textit{ext}_R|$.
\item{} $A \overset{p}{\longrightarrow} \lambda, i$ for a \textit{terminal} production where $E_R = \emptyset$ and $i = |V_R\backslash \textit{ext}_R|$.
\end{itemize}

Matrix $M_2$ (Tab. \ref{tab:PreProcessingTables}) shows the short-hand representation of the productions of the grammar in Figure \ref{fig:cnftggrammar}. Considering a second input $n$ as the size of the hypergraph to be generated, we are ready to describe a pair of algorithms $(\textbf{Pre},\textbf{Gen})$ for the random sampling of a hypergraph $H$ from a grammar $G$. Such a hypergraph is sampled in $L^A_n(G)$, where $A \in N$ is the non-terminal we begin the sampling from. If $A = S$ and $G$ is $n$-unambiguous, $H$ is sampled uniformly at random among all the hypergraphs in $L_n(G)$.

\subsection{Pre-processing phase}
\label{subsec:PreprocessingPhase}
The Pre-processing phase is used to construct a pair of matrices $M_1, M_2$ needed in the generation phase. Let $G = (N,\Sigma,P,S,(mark_{p})_{p \in P})$ be an $\textit{HRG}$, let $n \in \mathbb{N}$ be the size of the hypergraph $H \in L_n(G)$ we would like to generate, then the algorithm $\textbf{Pre}$ (Alg. \ref{alg:PreprocessingPhase}) produces the structures required for the generation.

\begin{table}[htpb]
\caption{Matrices $M_1$ and $M_2$ resulting from $\textbf{Pre}(G',12)$}
\label{tab:PreProcessingTables}
\begin{center}
\setlength\tabcolsep{3.5pt}
\begin{tabular}{cc}
$M_1$ & $M_2$\\\footnotesize
\begin{tabular}[t]{ l|r r r r r r r r r r r r }
 N & 1 & 2 & 3 & 4 & 5 & 6 & 7 & 8 & 9 & 10 & 11 & 12 \\
 \hline
 $A$ & 1 & 0 & 2 & 0 & 14 & 0 &  92 & 0 & 616 & 0 & 3920 & 0 \\
 $B$ & 2 & 0 & 8 & 0 & 32 & 0 & 128 & 0 & 256 & 0 &  512 & 0 \\
 $C$ & 0 & 0 & 2 & 0 & 32 & 0 &  76 & 0 & 488 & 0 & 2928 & 0 \\
 $D$ & 2 & 0 & 0 & 0 &  0 & 0 &   0 & 0 &   0 & 0 &    0 & 0 \\
 \hline
\end{tabular} &\footnotesize
\begin{tabular}[t]{ l|r r r r r r r r r r r r r }
 P \ \ \ \ \ \ \ \ \ \ \ \ \ \ \ \ \ & 1 & 2 & 3 & 4 & 5 & 6 & 7 & 8 & 9 & 10 & 11 & 12 \\
 \hline
 $A \overset{P1}{\longrightarrow} CA,1$ & 0 & 0 & 0 & 0 &  2 & 0 & 16 & 0 & 128 & 0 &  992 & 0 \\
 $A \overset{P2}{\longrightarrow} BA,1$ & 0 & 0 & 2 & 0 & 12 & 0 & 76 & 0 & 488 & 0 & 2928 & 0 \\
 $B \overset{P4}{\longrightarrow} DB,1$ & 0 & 0 & 4 & 0 & 16 & 0 & 64 & 0 & 128 & 0 &  256 & 0 \\
 $B \overset{P5}{\longrightarrow} DB,1$ & 0 & 0 & 4 & 0 & 16 & 0 & 64 & 0 & 128 & 0 &  256 & 0 \\
 $C \overset{P8}{\longrightarrow} BA,1$ & 0 & 0 & 2 & 0 & 12 & 0 & 76 & 0 & 488 & 0 & 2928 & 0 \\
 $A \overset{P3}{\longrightarrow} 1,0$  & 1 & 0 & 0 & 0 & 0  & 0 &  0 & 0 &   0 & 0 &    0 & 0 \\
 $B \overset{P6}{\longrightarrow} +,0$  & 1 & 0 & 0 & 0 & 0  & 0 &  0 & 0 &   0 & 0 &    0 & 0 \\
 $B \overset{P7}{\longrightarrow} *,0$  & 1 & 0 & 0 & 0 & 0  & 0 &  0 & 0 &   0 & 0 &    0 & 0 \\
 $D \overset{P9}{\longrightarrow} +,0$  & 1 & 0 & 0 & 0 & 0  & 0 &  0 & 0 &   0 & 0 &    0 & 0 \\
 $D \overset{P10}{\longrightarrow} *,0$ & 1 & 0 & 0 & 0 & 0  & 0 &  0 & 0 &   0 & 0 &    0 & 0 \\
 \hline
\end{tabular}
\end{tabular}
\end{center}
\end{table}

We begin initializing the entries of two matrices $M_1 = (N \times \mathbb{N})$ and $M_2 = (P \times \mathbb{N})$ to $0$. Each entry $(A, \ell)$ of $M_1$, also denoted as $A[\ell]$, represents the number of derivations yielding a hypergraph of size $\ell + \textit{type}(A)$, from a non-terminal $A \in N$. Each entry $(p, \ell)$ of $M_2$, also denoted as $p[\ell]$ represents the number of derivations yielding a hypergraph of size $\ell + |\textit{ext}_R|$, from a production $p \in P$. According to the type of production they are also denoted as $A \overset{p}{\longrightarrow} \lambda,i[\ell]$ or $A \overset{p}{\longrightarrow} a,i[\ell]$ for terminal productions and $A \overset{p}{\longrightarrow} BC,i[\ell]$ for a non-terminal production.
Considering each terminal production $p \in P_T$, either yielding a single terminal hyperedge $A \overset{p}{\longrightarrow} a,i$ or at least a single isolated node $A \overset{p}{\longrightarrow} \lambda,i$, the corresponding $M_2$ entry $p[i + 1]$ in the former case, or $p[i]$ in the latter, is set to $1$.
Then, for each $\ell \in \mathbb{N}$ in $1 \leq \ell \leq n$, for each non-terminal $A \in N$, $A[\ell] = \sum_{p \in P^A} p[\ell]$ and for each production $p \in P_N$, $p[\ell] = \sum_{0<k<\ell}B, [k] \cdot C[\ell - k]$.

The matrices can be used to generate hypergraphs in $L^A(G)$ of size $\ell + \textit{type}(A)$, with $1 \leq \ell \leq n$ from any non-terminal $A \in N$. If the non-terminal $A$ is chosen before the pre-processing phase we can reduce the size of the tables to $n - \textit{type}(A)$.
Table \ref{tab:PreProcessingTables} shows the result of running the algorithm $\textbf{Pre}$ using the grammar $G'$ in Figure \ref{fig:cnftggrammar} and a size of $12$ as input.

\subsection{Generation phase}
\label{subsec:GenerationPhase}
In the generation phase a non-terminal $\bar{A} \in N$ is chosen and a size-$\bar{n}$-hypergraph $H$, with $1 \leq \bar{n} \leq n + \textit{type}(A)$, is generated using the data collected in the matrices $M_1$, $M_2$ and a pseudo-random number generator $\textit{RNG}$. The algorithm $\textbf{Gen}$ (Alg. \ref{alg:GenerationPhase}) describes this process.

On input $\textbf{Gen}(G,\langle M_1, M_2 \rangle, \bar{A}, \bar{n} - \textit{type}(A))$, if $\bar{A}[\bar{n} - \textit{type}(A)] = 0$ the generating algorithm fails, otherwise, having $\bar{A}^{\bullet}$ as a basis, the algorithm recursively calls the function $\textbf{derH}$ proceeding through the following steps:

\begin{enumerate}
\setlength\itemsep{-0.2em}
\item The $\textit{RNG}$ is used to choose a production $p \in P^A$ with probability $p[\ell]/A[\ell]$.

\item If $p \in P^A_{\Sigma}$, the replacement of $e$, the $\textit{handle}$ of $A$, with the hypergraph $R$ in $\textit{rhs}(p)$ is returned.

\item If $p \in P^A_T$ the $\textit{RNG}$ is used again to choose a ``split" $0 < k < \ell'$ with $\ell' = \ell - i$ and probability $B[k] \cdot C[\ell' - k] / A \overset{p}{\longrightarrow} BC,i[\ell]$. The hypergraph $\textit{rhs}(p)[e_{\alpha} / \textbf{derH}(B, k),e_{\beta} / \textbf{derH}(C,\ell' - k)]$ produced by the replacement of $e_{\alpha}$ with the result on the recursive function on input $\textbf{derH}(B, k)$ and the replacement of $e_{\beta}$ with the result of the recursive function on input $\textbf{derH}(C,\ell' - k)$ is computed. Then, the replacement of the hyperedge $e$, the $\textit{handle}$ of $A$, with the aforementioned hypergraph is returned. We use the notation $B_k C_{\ell'-k}$ to indicate such a split.
\end{enumerate}

The derivation $d = A^{\bullet} \Rightarrow^*_P H$ in Figure \ref{fig:termsgraphderivation} corresponds to the sequence of replacements computed by the recursive function $\textbf{derH}$ to generate the size-$12$-hypergraph $H$ in Figure \ref{fig:hypergraph}, using non-terminal $A$ as input. For each step we show the probability of the production $p$ to be chosen and the choice of the split and its probability if $p \in P^N$. Since $G$ is non-ambiguous, the first step shows that $|L_{12}(G)| = 3920$, that is, there are $3920$ unique size-$12$-hypergraphs to choose from, each having a different ordered derivation tree. Figure \ref{fig:tree} shows the tree $t$ for which $\textit{yield}(t) = H$, so that $\textit{trav}(t)$, or equivalently $\textit{lmd}(H)$, corresponds to the unique sequence of productions applied by the generation algorithm to produce $H$. In the figure are also indicated the starting symbol $A$ and the replaced hyperedges $e_{\alpha}$ and $e_{\beta}$, respectively on the edges connecting the left and right child of each node. The proof of termination of the Generation algorithm is based on the assumption that the input grammar is non-contracting:

\begin{algorithm}[htpb]\label{alg:PreprocessingPhase}
\begin{multicols}{2}
\setstretch{0.65}
\caption{$\textbf{Pre}$ - Pre-processing phase}
\SetAlgoLined
 \textbf{Input:} $(G, n)$, where $G = (N,\Sigma,P,S,(mark_{p})_{p \in P})$ and $n \in \mathbb{N}$, $n \geq 1$\\
 \textbf{Output:} $\langle M_1, M_2 \rangle$\\
 \hrulefill\\
 \For{$1 \leq \ell \leq n$}{
  \ForEach{$A \in N$}{
   $A[\ell] := 0$\;
  }
  \ForEach{$p \in P$}{
   $p[\ell] := 0$\;
  }
 }
 \ForEach{$A \overset{p}{\longrightarrow} a,i \in P_{\Sigma}$}{
  $A \overset{p}{\longrightarrow} a,i[i + 1] := 1$\;
 }
 \ForEach{$A \overset{p}{\longrightarrow} \lambda,i \in P_{\Sigma}$}{
  $A \overset{p}{\longrightarrow} \lambda,i[i] := 1$\;
 }
 \For{$1 \leq \ell \leq n$}{
  \ForEach{$A \in N$}{
   \ForEach{$p \in P^A$}{
    $A[\ell] := A[\ell] + p[\ell]$\;
   }
  }
  \ForEach{$A \overset{p}{\longrightarrow} BC,i \in P_N$}{
   \For{$1 \leq k < \ell$}{
    $A \overset{p}{\longrightarrow} BC,i[\ell + i] := A \overset{p}{\longrightarrow} BC,i[\ell + i] + B[k] \cdot C[\ell-k]$\;
   }
  }
 }
\end{multicols}
\end{algorithm}

\begin{algorithm}[htpb]\label{alg:GenerationPhase}
\setstretch{0.65}
\caption{$\textbf{Gen}$ - Generation phase}
\SetAlgoLined
 \textbf{Input:} $(G, \langle M_1,M_2 \rangle, \bar{A}, \bar{n})$, where $G = (N,\Sigma,P,S,(mark_{p})_{p \in P})$, $\langle M_1,M_2 \rangle := \textbf{Pre}(G, n)$, $\bar{A} \in N$ and $\bar{n} \in \mathbb{N}$, $1 \leq \bar{n} \leq n + \textit{type}(\bar{A})$\\
 \textbf{Output:} $H \in L^{\bar{A}}_{\bar{n}}(G)$\\
 \hrulefill\\
 $\ell = \bar{n} - \textit{type}(\bar{A})$\\
 \If{$\bar{A}[\ell] = 0$}{
  \Return $\bot$\;
 }
 Recursively generate $H$ using $(\bar{A}, \ell)$ as first input as follows:\\
 \SetKwFunction{FMain}{}
 \SetKwProg{Fn}{function derH}{:}{}
 \Fn{\FMain{$A$, $\ell$}}{
  $p \longleftarrow \textit{RNG}$ with $p \in P^A$ and probability $p[\ell] / A[\ell]$\;
  \eIf{$p \in P_T$}{
   \Return $A^{\bullet}[e / R]$\;  
  }{
   $\ell' = \ell - i$\;
   $k \longleftarrow \textit{RNG}$ with $0 < k < \ell'$ and probability $B[k] \cdot C[{\ell' - k}] / (A \overset{p}{\longrightarrow} BC,i)[{\ell}]$\;
   \Return $A^{\bullet}[e / R[e_\alpha / \textbf{derH}(B, k),e_\beta / \textbf{derH}(C, \ell' - k)]]$\;
  }
 }
 \textbf{end function}
\end{algorithm}

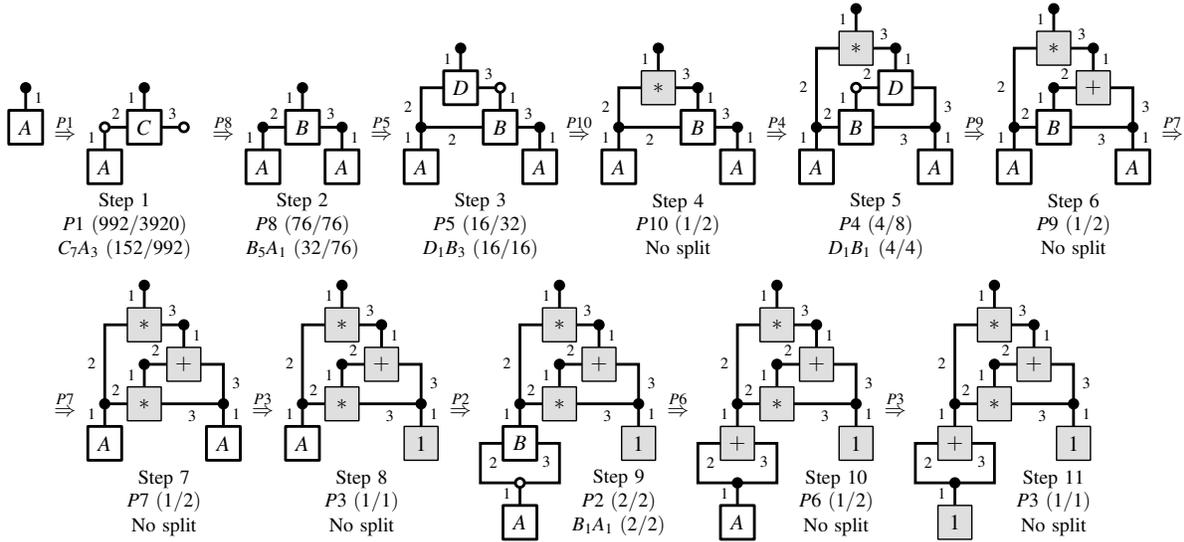
\begin{figure}[htpb]
\centering
\resizebox{\textwidth}{!}{%
\begin{tikzpicture}
 \begin{scope} [local bounding box = c1, shift = {(0.0, 0.0)}]
  \draw [ue] (0.0, 0.3) -- (0.0, 0.7) node[el, midway, right] {1} ;
  \draw [hl] (-0.3, 0.3) rectangle (0.3, -0.3) node[midway] {$A$} ;
  \path (0.0, 0.7) pic{ev} ;
 \end{scope}
 
 \node[text width = 1.0cm, align=center] at (0.7, 0.0) {$\overset{P1}{\Rightarrow}$} ;
 
 \begin{scope} [local bounding box = c1, shift = {(2.1, 0.0)}]
  \draw [ue] (0.0, 0.3) -- (0.0, 0.7) node[el, midway, left] {1} ;
  \draw [ue] (-0.3, 0.0) -- (-0.7, 0.0) node[el, midway, above] {2} ;
  \draw [ue] (0.3, 0.0) -- (0.7, 0.0) node[el, midway, above] {3} ;
  \draw [ue] (-0.7, -0.4) -- (-0.7, 0.0) node[el, midway, left] {1} ;
  \draw [hl] (-0.3, 0.3) rectangle (0.3, -0.3) node[midway] {$C$} ;
  \draw [hl] (-1.0, -0.4) rectangle (-0.4, -1.0) node[midway] {$A$} ;
  \path (0.0, 0.7) pic{ev} ;
  \path (-0.7, 0.0) pic{iv} ;
  \path (0.7, 0.0) pic{iv} ;
  \node[text width = 4.0cm, font = \small, align=center] at (-0.35, -1.75) {Step 1\\$P1$ $(992/3920)$\\$C_7 A_3$ $(152/992)$} ;
 \end{scope}
 
 \node[text width = 1.0cm, align=center] at (3.5, 0.0) {$\overset{P8}{\Rightarrow}$} ;
 
 \begin{scope} [local bounding box = c1, shift = {(4.9, 0.0)}]
  \draw [ue] (-0.7, -0.4) -- (-0.7, 0.0) node[el, midway, left] {1} ;
  \draw [ue] (0.0, 0.3) -- (0.0, 0.7) node[el, midway, left] {1} ;
  \draw [ue] (-0.3, 0.0) -- (-0.7, 0.0) node[el, midway, above] {2} ;
  \draw [ue] (0.3, 0.0) -- (0.7, 0.0) node[el, midway, above] {3} ;
  \draw [ue] (0.7, -0.4) -- (0.7, 0.0) node[el, midway, right] {1} ;
  \draw [hl] (-1.0, -0.4) rectangle (-0.4, -1.0) node[midway] {$A$} ;
  \draw [hl] (-0.3, 0.3) rectangle (0.3, -0.3) node[midway] {$B$} ;
  \draw [hl] (0.4, -0.4) rectangle (1.0, -1.0) node[midway] {$A$} ;
  \path (0.0, 0.7) pic{ev} ;
  \path (-0.7, 0.0) pic{ev} ;
  \path (0.7, 0.0) pic{ev} ;
  \node[text width = 4.0cm, font = \small, align=center] at (0.0, -1.75) {Step 2\\$P8$ $(76/76)$\\$B_5 A_1$ $(32/76)$} ;
 \end{scope}
 
 \node[text width = 1.0cm, align=center] at (6.3, 0.0) {$\overset{P5}{\Rightarrow}$} ;
 
 \begin{scope} [local bounding box = c1, shift = {(7.7, 0.0)}]
  \draw [ue] (-0.7, -0.4) -- (-0.7, 0.0) node[el, midway, left] {1} ;
  \draw [ue] (0.0, 1.0) -- (0.0, 1.4) node[el, midway, left] {1} ;
  \draw [ue] (-0.3, 0.7) -- (-0.7, 0.7) -- (-0.7, 0.0) node[el, midway, left] {2} ;
  \draw [ue] (0.3, 0.7) -- (0.7, 0.7) node[el, midway, above] {3} ;
  \draw [ue] (0.7, 0.3) -- (0.7, 0.7) node[el, midway, right] {1} ;
  \draw [ue] (0.4, 0.0) -- (-0.7, 0.0) node[el, midway, below] {2} ;
  \draw [ue] (1.0, 0.0) -- (1.4, 0.0) node[el, midway, above] {3} ;
  \draw [ue] (1.4, -0.4) -- (1.4, 0.0) node[el, midway, right] {1} ;
  \draw [hl] (-1.0, -0.4) rectangle (-0.4, -1.0) node[midway] {$A$} ;
  \draw [hl] (0.4, 0.3) rectangle (1.0, -0.3) node[midway] {$B$} ;
  \draw [hl] (1.1, -0.4) rectangle (1.7, -1.0) node[midway] {$A$} ;
  \draw [hl] (-0.3, 1.0) rectangle (0.3, 0.4) node[midway] {$D$} ;
  \path (0.0, 1.4) pic{ev} ;
  \path (-0.7, 0.0) pic{ev} ;
  \path (1.4, 0.0) pic{ev} ;
  \path (0.7, 0.7) pic{iv} ;
  \node[text width = 4.0cm, font = \small, align=center] at (0.35, -1.75) {Step 3\\$P5$ $(16/32)$\\$D_1 B_3$ $(16/16)$} ;
 \end{scope}
 
 \node[text width = 1.0cm, align=center] at (9.8, 0.0) {$\overset{P10}{\Rightarrow}$} ;
 
 \begin{scope} [local bounding box = c1, shift = {(11.2, 0.0)}]
  \draw [ue] (-0.7, -0.4) -- (-0.7, 0.0) node[el, midway, left] {1} ;
  \draw [ue] (0.0, 1.0) -- (0.0, 1.4) node[el, midway, left] {1} ;
  \draw [ue] (-0.3, 0.7) -- (-0.7, 0.7) -- (-0.7, 0.0) node[el, midway, left] {2} ;
  \draw [ue] (0.3, 0.7) -- (0.7, 0.7) node[el, midway, above] {3} ;
  \draw [ue] (0.7, 0.3) -- (0.7, 0.7) node[el, midway, right] {1} ;
  \draw [ue] (0.4, 0.0) -- (-0.7, 0.0) node[el, midway, below] {2} ;
  \draw [ue] (1.0, 0.0) -- (1.4, 0.0) node[el, midway, above] {3} ;
  \draw [ue] (1.4, -0.4) -- (1.4, 0.0) node[el, midway, right] {1} ;
  \draw [hl] (-1.0, -0.4) rectangle (-0.4, -1.0) node[midway] {$A$} ;
  \draw [hl] (0.4, 0.3) rectangle (1.0, -0.3) node[midway] {$B$} ;
  \draw [hl] (1.1, -0.4) rectangle (1.7, -1.0) node[midway] {$A$} ;
  \draw [ht] (-0.3, 1.0) rectangle (0.3, 0.4) node[midway] {$*$} ;
  \path (0.0, 1.4) pic{ev} ;
  \path (-0.7, 0.0) pic{ev} ;
  \path (1.4, 0.0) pic{ev} ;
  \path (0.7, 0.7) pic{ev} ;
  \node[text width = 4.0cm, font = \small, align=center] at (0.35, -1.75) {Step 4\\$P10$ $(1/2)$\\No split} ;
 \end{scope}
 
 \node[text width = 1.0cm, align=center] at (13.3, 0.0) {$\overset{P4}{\Rightarrow}$} ;
 
 \begin{scope} [local bounding box = c1, shift = {(14.7, 0.0)}]
  \draw [ue] (-0.7, -0.4) -- (-0.7, 0.0) node[el, midway, left] {1} ;
  \draw [ue] (0.0, 1.7) -- (0.0, 2.1) node[el, midway, left] {1} ;
  \draw [ue] (-0.3, 1.4) -- (-0.7, 1.4) -- (-0.7, 0.0) node[el, midway, left] {2} ;
  \draw [ue] (0.3, 1.4) -- (0.7, 1.4) node[el, midway, above] {3} ;
  \draw [ue] (0.0, 0.3) -- (0.0, 0.7) node[el, midway, left] {1} ;
  \draw [ue] (-0.3, 0.0) -- (-0.7, 0.0) node[el, midway, above] {2} ;
  \draw [ue] (0.3, 0.0) -- (1.4, 0.0) node[el, midway, below] {3} ;
  \draw [ue] (0.7, 1.0) -- (0.7, 1.4) node[el, midway, right] {1} ;
  \draw [ue] (0.4, 0.7) -- (0.0, 0.7) node[el, midway, above] {2} ;
  \draw [ue] (1.0, 0.7) -- (1.4, 0.7) -- (1.4, 0.0) node[el, midway, right] {3} ;
  \draw [ue] (1.4, -0.4) -- (1.4, 0.0) node[el, midway, right] {1} ;
  \draw [hl] (-1.0, -0.4) rectangle (-0.4, -1.0) node[midway] {$A$} ;
  \draw [hl] (-0.3, 0.3) rectangle (0.3, -0.3) node[midway] {$B$} ;
  \draw [hl] (1.1, -0.4) rectangle (1.7, -1.0) node[midway] {$A$} ;
  \draw [ht] (-0.3, 1.7) rectangle (0.3, 1.1) node[midway] {$*$} ;
  \draw [hl] (0.4, 1.0) rectangle (1.0, 0.4) node[midway] {$D$} ;
  \path (0.0, 2.1) pic{ev} ;
  \path (-0.7, 0.0) pic{ev} ;
  \path (1.4, 0.0) pic{ev} ;
  \path (0.7, 1.4) pic{ev} ;
  \path (0.0, 0.7) pic{iv} ;
  \node[text width = 4.0cm, font = \small, align=center] at (0.35, -1.75) {Step 5\\$P4$ $(4/8)$\\$D_1 B_1$ $(4/4)$} ;
 \end{scope}
 
 \node[text width = 1.0cm, align=center] at (16.8, 0.0) {$\overset{P9}{\Rightarrow}$} ;
 
 \begin{scope} [local bounding box = c1, shift = {(18.2, 0.0)}]
  \draw [ue] (-0.7, -0.4) -- (-0.7, 0.0) node[el, midway, left] {1} ;
  \draw [ue] (0.0, 1.7) -- (0.0, 2.1) node[el, midway, left] {1} ;
  \draw [ue] (-0.3, 1.4) -- (-0.7, 1.4) -- (-0.7, 0.0) node[el, midway, left] {2} ;
  \draw [ue] (0.3, 1.4) -- (0.7, 1.4) node[el, midway, above] {3} ;
  \draw [ue] (0.0, 0.3) -- (0.0, 0.7) node[el, midway, left] {1} ;
  \draw [ue] (-0.3, 0.0) -- (-0.7, 0.0) node[el, midway, above] {2} ;
  \draw [ue] (0.3, 0.0) -- (1.4, 0.0) node[el, midway, below] {3} ;
  \draw [ue] (0.7, 1.0) -- (0.7, 1.4) node[el, midway, right] {1} ;
  \draw [ue] (0.4, 0.7) -- (0.0, 0.7) node[el, midway, above] {2} ;
  \draw [ue] (1.0, 0.7) -- (1.4, 0.7) -- (1.4, 0.0) node[el, midway, right] {3} ;
  \draw [ue] (1.4, -0.4) -- (1.4, 0.0) node[el, midway, right] {1} ;
  \draw [hl] (-1.0, -0.4) rectangle (-0.4, -1.0) node[midway] {$A$} ;
  \draw [hl] (-0.3, 0.3) rectangle (0.3, -0.3) node[midway] {$B$} ;
  \draw [hl] (1.1, -0.4) rectangle (1.7, -1.0) node[midway] {$A$} ;
  \draw [ht] (-0.3, 1.7) rectangle (0.3, 1.1) node[midway] {$*$} ;
  \draw [ht] (0.4, 1.0) rectangle (1.0, 0.4) node[midway] {$+$} ;
  \path (0.0, 2.1) pic{ev} ;
  \path (-0.7, 0.0) pic{ev} ;
  \path (1.4, 0.0) pic{ev} ;
  \path (0.7, 1.4) pic{ev} ;
  \path (0.0, 0.7) pic{ev} ;
  \node[text width = 4.0cm, font = \small, align=center] at (0.35, -1.75) {Step 6\\$P9$ $(1/2)$\\No split} ;
 \end{scope}
 
 \node[text width = 1.0cm, align=center] at (20.3, 0.0) {$\overset{P7}{\Rightarrow}$} ;
 
 \node[text width = 1.0cm, align=center] at (0.7, -4.9) {$\overset{P7}{\Rightarrow}$} ;
 
 \begin{scope} [local bounding box = c1, shift = {(2.1, -4.9)}]
  \draw [ue] (-0.7, -0.4) -- (-0.7, 0.0) node[el, midway, left] {1} ;
  \draw [ue] (0.0, 1.7) -- (0.0, 2.1) node[el, midway, left] {1} ;
  \draw [ue] (-0.3, 1.4) -- (-0.7, 1.4) -- (-0.7, 0.0) node[el, midway, left] {2} ;
  \draw [ue] (0.3, 1.4) -- (0.7, 1.4) node[el, midway, above] {3} ;
  \draw [ue] (0.0, 0.3) -- (0.0, 0.7) node[el, midway, left] {1} ;
  \draw [ue] (-0.3, 0.0) -- (-0.7, 0.0) node[el, midway, above] {2} ;
  \draw [ue] (0.3, 0.0) -- (1.4, 0.0) node[el, midway, below] {3} ;
  \draw [ue] (0.7, 1.0) -- (0.7, 1.4) node[el, midway, right] {1} ;
  \draw [ue] (0.4, 0.7) -- (0.0, 0.7) node[el, midway, above] {2} ;
  \draw [ue] (1.0, 0.7) -- (1.4, 0.7) -- (1.4, 0.0) node[el, midway, right] {3} ;
  \draw [ue] (1.4, -0.4) -- (1.4, 0.0) node[el, midway, right] {1} ;
  \draw [hl] (-1.0, -0.4) rectangle (-0.4, -1.0) node[midway] {$A$} ;
  \draw [ht] (-0.3, 0.3) rectangle (0.3, -0.3) node[midway] {$*$} ;
  \draw [hl] (1.1, -0.4) rectangle (1.7, -1.0) node[midway] {$A$} ;
  \draw [ht] (-0.3, 1.7) rectangle (0.3, 1.1) node[midway] {$*$} ;
  \draw [ht] (0.4, 1.0) rectangle (1.0, 0.4) node[midway] {$+$} ;
  \path (0.0, 2.1) pic{ev} ;
  \path (-0.7, 0.0) pic{ev} ;
  \path (1.4, 0.0) pic{ev} ;
  \path (0.7, 1.4) pic{ev} ;
  \path (0.0, 0.7) pic{ev} ;
  \node[text width = 4.0cm, font = \small, align=center] at (0.35, -1.75) {Step 7\\$P7$ $(1/2)$\\No split} ;
 \end{scope}
 
 \node[text width = 1.0cm, align=center] at (4.2, -4.9) {$\overset{P3}{\Rightarrow}$} ;
 
 \begin{scope} [local bounding box = c1, shift = {(5.6, -4.9)}]
  \draw [ue] (-0.7, -0.4) -- (-0.7, 0.0) node[el, midway, left] {1} ;
  \draw [ue] (0.0, 1.7) -- (0.0, 2.1) node[el, midway, left] {1} ;
  \draw [ue] (-0.3, 1.4) -- (-0.7, 1.4) -- (-0.7, 0.0) node[el, midway, left] {2} ;
  \draw [ue] (0.3, 1.4) -- (0.7, 1.4) node[el, midway, above] {3} ;
  \draw [ue] (0.0, 0.3) -- (0.0, 0.7) node[el, midway, left] {1} ;
  \draw [ue] (-0.3, 0.0) -- (-0.7, 0.0) node[el, midway, above] {2} ;
  \draw [ue] (0.3, 0.0) -- (1.4, 0.0) node[el, midway, below] {3} ;
  \draw [ue] (0.7, 1.0) -- (0.7, 1.4) node[el, midway, right] {1} ;
  \draw [ue] (0.4, 0.7) -- (0.0, 0.7) node[el, midway, above] {2} ;
  \draw [ue] (1.0, 0.7) -- (1.4, 0.7) -- (1.4, 0.0) node[el, midway, right] {3} ;
  \draw [ue] (1.4, -0.4) -- (1.4, 0.0) node[el, midway, right] {1} ;
  \draw [hl] (-1.0, -0.4) rectangle (-0.4, -1.0) node[midway] {$A$} ;
  \draw [ht] (-0.3, 0.3) rectangle (0.3, -0.3) node[midway] {$*$} ;
  \draw [ht] (1.1, -0.4) rectangle (1.7, -1.0) node[midway] {$1$} ;
  \draw [ht] (-0.3, 1.7) rectangle (0.3, 1.1) node[midway] {$*$} ;
  \draw [ht] (0.4, 1.0) rectangle (1.0, 0.4) node[midway] {$+$} ;
  \path (0.0, 2.1) pic{ev} ;
  \path (-0.7, 0.0) pic{ev} ;
  \path (1.4, 0.0) pic{ev} ;
  \path (0.7, 1.4) pic{ev} ;
  \path (0.0, 0.7) pic{ev} ;
  \node[text width = 4.0cm, font = \small, align=center] at (0.35, -1.75) {Step 8\\$P3$ $(1/1)$\\No split} ;
 \end{scope}
 
 \node[text width = 1.0cm, align=center] at (7.7, -4.9) {$\overset{P2}{\Rightarrow}$} ;
 
 \begin{scope} [local bounding box = c1, shift = {(9.45, -4.9)}]
  \draw [ue] (-0.7, -0.4) -- (-0.7, 0.0) node[el, midway, left] {1} ;
  \draw [ue] (-1.0, -0.7) -- (-1.4, -0.7) -- (-1.4, -1.4) node[el, midway, right] {2} -- (-0.7, -1.4) ;
  \draw [ue] (-0.4, -0.7) -- (0.0, -0.7) -- (0.0, -1.4) node[el, midway, left] {3} -- (-0.7, -1.4) ;
  \draw [ue] (0.0, 1.7) -- (0.0, 2.1) node[el, midway, left] {1} ;
  \draw [ue] (-0.3, 1.4) -- (-0.7, 1.4) -- (-0.7, 0.0) node[el, midway, left] {2} ;
  \draw [ue] (0.3, 1.4) -- (0.7, 1.4) node[el, midway, above] {3} ;
  \draw [ue] (0.0, 0.3) -- (0.0, 0.7) node[el, midway, left] {1} ;
  \draw [ue] (-0.3, 0.0) -- (-0.7, 0.0) node[el, midway, above] {2} ;
  \draw [ue] (0.3, 0.0) -- (1.4, 0.0) node[el, midway, below] {3} ;
  \draw [ue] (0.7, 1.0) -- (0.7, 1.4) node[el, midway, right] {1} ;
  \draw [ue] (0.4, 0.7) -- (0.0, 0.7) node[el, midway, above] {2} ;
  \draw [ue] (1.0, 0.7) -- (1.4, 0.7) -- (1.4, 0.0) node[el, midway, right] {3} ;
  \draw [ue] (1.4, -0.4) -- (1.4, 0.0) node[el, midway, right] {1} ;
  \draw [ue] (-0.7, -1.8) -- (-0.7, -1.4) node[el, midway, left] {1} ;
  \draw [hl] (-1.0, -0.4) rectangle (-0.4, -1.0) node[midway] {$B$} ;
  \draw [ht] (-0.3, 0.3) rectangle (0.3, -0.3) node[midway] {$*$} ;
  \draw [ht] (1.1, -0.4) rectangle (1.7, -1.0) node[midway] {$1$} ;
  \draw [ht] (-0.3, 1.7) rectangle (0.3, 1.1) node[midway] {$*$} ;
  \draw [ht] (0.4, 1.0) rectangle (1.0, 0.4) node[midway] {$+$} ;
  \draw [hl] (-1.0, -1.8) rectangle (-0.4, -2.4) node[midway] {$A$} ;
  \path (0.0, 2.1) pic{ev} ;
  \path (-0.7, 0.0) pic{ev} ;
  \path (1.4, 0.0) pic{ev} ;
  \path (0.7, 1.4) pic{ev} ;
  \path (0.0, 0.7) pic{ev} ;
  \path (-0.7, -1.4) pic{iv} ;
  \node[text width = 4.0cm, font = \small, align=center] at (1.05, -1.75) {Step 9\\$P2$ $(2/2)$\\$B_1 A_1$ $(2/2)$} ;
 \end{scope}
 
 \node[text width = 1.0cm, align=center] at (11.55, -4.9) {$\overset{P6}{\Rightarrow}$} ;
 
 \begin{scope} [local bounding box = c1, shift = {(13.3, -4.9)}]
  \draw [ue] (-0.7, -0.4) -- (-0.7, 0.0) node[el, midway, left] {1} ;
  \draw [ue] (-1.0, -0.7) -- (-1.4, -0.7) -- (-1.4, -1.4) node[el, midway, right] {2} -- (-0.7, -1.4) ;
  \draw [ue] (-0.4, -0.7) -- (0.0, -0.7) -- (0.0, -1.4) node[el, midway, left] {3} -- (-0.7, -1.4) ;
  \draw [ue] (0.0, 1.7) -- (0.0, 2.1) node[el, midway, left] {1} ;
  \draw [ue] (-0.3, 1.4) -- (-0.7, 1.4) -- (-0.7, 0.0) node[el, midway, left] {2} ;
  \draw [ue] (0.3, 1.4) -- (0.7, 1.4) node[el, midway, above] {3} ;
  \draw [ue] (0.0, 0.3) -- (0.0, 0.7) node[el, midway, left] {1} ;
  \draw [ue] (-0.3, 0.0) -- (-0.7, 0.0) node[el, midway, above] {2} ;
  \draw [ue] (0.3, 0.0) -- (1.4, 0.0) node[el, midway, below] {3} ;
  \draw [ue] (0.7, 1.0) -- (0.7, 1.4) node[el, midway, right] {1} ;
  \draw [ue] (0.4, 0.7) -- (0.0, 0.7) node[el, midway, above] {2} ;
  \draw [ue] (1.0, 0.7) -- (1.4, 0.7) -- (1.4, 0.0) node[el, midway, right] {3} ;
  \draw [ue] (1.4, -0.4) -- (1.4, 0.0) node[el, midway, right] {1} ;
  \draw [ue] (-0.7, -1.8) -- (-0.7, -1.4) node[el, midway, left] {1} ;
  \draw [ht] (-1.0, -0.4) rectangle (-0.4, -1.0) node[midway] {$+$} ;
  \draw [ht] (-0.3, 0.3) rectangle (0.3, -0.3) node[midway] {$*$} ;
  \draw [ht] (1.1, -0.4) rectangle (1.7, -1.0) node[midway] {$1$} ;
  \draw [ht] (-0.3, 1.7) rectangle (0.3, 1.1) node[midway] {$*$} ;
  \draw [ht] (0.4, 1.0) rectangle (1.0, 0.4) node[midway] {$+$} ;
  \draw [hl] (-1.0, -1.8) rectangle (-0.4, -2.4) node[midway] {$A$} ;
  \path (0.0, 2.1) pic{ev} ;
  \path (-0.7, 0.0) pic{ev} ;
  \path (1.4, 0.0) pic{ev} ;
  \path (0.7, 1.4) pic{ev} ;
  \path (0.0, 0.7) pic{ev} ;
  \path (-0.7, -1.4) pic{ev} ;
  \node[text width = 4.0cm, font = \small, align=center] at (1.05, -1.75) {Step 10\\$P6$ $(1/2)$\\No split} ;
 \end{scope}
 
 \node[text width = 1.0cm, align=center] at (15.4, -4.9) {$\overset{P3}{\Rightarrow}$} ;
 
 \begin{scope} [local bounding box = c1, shift = {(17.15, -4.9)}]
  \draw [ue] (-0.7, -0.4) -- (-0.7, 0.0) node[el, midway, left] {1} ;
  \draw [ue] (-1.0, -0.7) -- (-1.4, -0.7) -- (-1.4, -1.4) node[el, midway, right] {2} -- (-0.7, -1.4) ;
  \draw [ue] (-0.4, -0.7) -- (0.0, -0.7) -- (0.0, -1.4) node[el, midway, left] {3} -- (-0.7, -1.4) ;
  \draw [ue] (0.0, 1.7) -- (0.0, 2.1) node[el, midway, left] {1} ;
  \draw [ue] (-0.3, 1.4) -- (-0.7, 1.4) -- (-0.7, 0.0) node[el, midway, left] {2} ;
  \draw [ue] (0.3, 1.4) -- (0.7, 1.4) node[el, midway, above] {3} ;
  \draw [ue] (0.0, 0.3) -- (0.0, 0.7) node[el, midway, left] {1} ;
  \draw [ue] (-0.3, 0.0) -- (-0.7, 0.0) node[el, midway, above] {2} ;
  \draw [ue] (0.3, 0.0) -- (1.4, 0.0) node[el, midway, below] {3} ;
  \draw [ue] (0.7, 1.0) -- (0.7, 1.4) node[el, midway, right] {1} ;
  \draw [ue] (0.4, 0.7) -- (0.0, 0.7) node[el, midway, above] {2} ;
  \draw [ue] (1.0, 0.7) -- (1.4, 0.7) -- (1.4, 0.0) node[el, midway, right] {3} ;
  \draw [ue] (1.4, -0.4) -- (1.4, 0.0) node[el, midway, right] {1} ;
  \draw [ue] (-0.7, -1.8) -- (-0.7, -1.4) node[el, midway, left] {1} ;
  \draw [ht] (-1.0, -0.4) rectangle (-0.4, -1.0) node[midway] {$+$} ;
  \draw [ht] (-0.3, 0.3) rectangle (0.3, -0.3) node[midway] {$*$} ;
  \draw [ht] (1.1, -0.4) rectangle (1.7, -1.0) node[midway] {$1$} ;
  \draw [ht] (-0.3, 1.7) rectangle (0.3, 1.1) node[midway] {$*$} ;
  \draw [ht] (0.4, 1.0) rectangle (1.0, 0.4) node[midway] {$+$} ;
  \draw [ht] (-1.0, -1.8) rectangle (-0.4, -2.4) node[midway] {$1$} ;
  \path (0.0, 2.1) pic{ev} ;
  \path (-0.7, 0.0) pic{ev} ;
  \path (1.4, 0.0) pic{ev} ;
  \path (0.7, 1.4) pic{ev} ;
  \path (0.0, 0.7) pic{ev} ;
  \path (-0.7, -1.4) pic{ev} ;
  \node[text width = 4.0cm, font = \small, align=center] at (1.05, -1.75) {Step 11\\$P3$ $(1/1)$\\No split} ;
 \end{scope}
\end{tikzpicture}%
}
\caption{A derivation $d = A^{\bullet} \Rightarrow^*_P H$ using the grammar of Figure \ref{fig:cnftggrammar}}
\label{fig:termsgraphderivation}
\end{figure}

\begin{proof}
Let's consider a measure equivalent to the size of a hypergraph $|H|$. To each application of the recursive function $\textbf{derH}$ in each step of the algorithm $\textbf{Gen}$, corresponds a direct derivation between two sentential forms $F \Rightarrow F'$ such that $F \leq F'$.
Since the grammar is in $\textit{CNF}$, at each step there are two possible cases:

\begin{enumerate}
\item[1.] $\textbf{derH}$ chooses a non-terminal production. In this case a single hyperedge $e \in F$ is replaced with a hypergraph $R \subseteq F'$ containing $2$ hyperedges and $0$ or more internal nodes. Clearly $|F| < |F'|$, meaning that the size of the sentential forms gets progressively close to $n$.

\item[2.] $\textbf{derH}$ chooses a terminal production. A hyperedge is replaced by a terminal hyperedge or a single node and $0$ or more additional internal nodes. In this case $|F| \leq |F'|$. Even if the size is not incremented, being a terminal production, the recursion does not progress any further.
\end{enumerate}

If it is not possible to generate a size-$n$-hypergraph using the input grammar $G$ the algorithm trivially ends in one step.
\end{proof}

\section{Uniform distribution and time complexity}
\label{sec:uniformcomplexity}

We now state our first main result, the uniform generation guarantee for Algorithm \ref{alg:GenerationPhase}.

\begin{theorem}\label{the:UniformSampling}
Given a grammar $G = (N,\Sigma,P,S,(mark_{p})_{p \in P})$, Algorithm \ref{alg:GenerationPhase} generates from every non-terminal $A \in N$ a size-$n$-hypergraph $H \in L^A_n(G)$, provided that $L^A_n(G) \neq \emptyset$. If $G$ is $n$-unambiguous and $\textit{RNG}$ is a uniform random number generator, the hypergraph is chosen uniformly at random.
\end{theorem}

\begin{proof}
Let $G$ be an $n$-unambiguous grammar in \textit{CNF}, the recursive function $\textbf{derH}$ derives a hypergraph $H \in L^{\bar{A}}_{\bar{n}}(G)$ simulating $\textit{trav}(t)$ where $\textit{yield}(t) = H$ and let $P(c_j)$ denote the probability of the $j$th choice $c$ made using the $\textit{RNG}$ at each step of the recursion, for a production or a split, according to $\textit{lmd}(H)$.

Let's recall that for the parallelization, confluence and associativity properties of context-free hyperedge replacement grammars \cite{courcelle-1987-aad}, the sequence of replacements associated to a derivation preserves the result of the derivation, despite of the order in which the replacements are applied. Thus, we are able to discuss each of its steps independently.

By definition, since the grammar is $n$-unambiguous, for any non-terminal $A \in N$ we know that the set of hypergraphs that can be generated using different productions $p \in P^A$ are pairwise distinct. Otherwise, there would exist $\textit{trav}(t') \neq \textit{trav}(t'')$ for which $\textit{yield}(t') \cong \textit{yield}(t'')$.

From algorithm $\textbf{Pre}$ (Alg. \ref{alg:PreprocessingPhase}) we know that $\sum_{p \in P^A} p[\ell] = A[\ell]$ and so the probability of the choice $c_j$ of each production in $\textit{lmd}(H)$ can be expressed by $P(c_j) = p[\ell]/A[\ell]$. Also, if $p \in P_N$, since the grammar is $n$-unambiguous the subsets of hypergraphs that can be derived by choosing different splits are also pairwise distinct. For a production $p \in P_N$ then $\sum_{0 < k < \ell} B[k] \cdot C[\ell' - k] = A \overset{p}{\longrightarrow} BC,i[\ell]$, thus a split can be chosen with probability $P(c_j) = B[k] \cdot C[\ell' - k] / p[\ell]$.

Knowing that for an $\textit{lmd}$, if the grammar is $n$-unambiguous, both the choices of productions and splits are made from independent sets, considering the corresponding derivation tree $t$, the probabilities associated to the choice of a node $P(c)$ and the ones associated to its children $P(c')$ and $P(c'')$ are of the form $\frac{m}{q}$, $\frac{m'}{q'}$ and $\frac{m''}{q''}$ with $m,m',m'',q,q',q'' \in \mathbb{N}$ and $q'q'' = m$. Moreover, the probabilities of two consecutive choices $P(c)$ and $P(c')$ are bound to the law of compound probabilities \cite{laplace-1820-tad}, that is, the choice of a node given the choice of its parent is of the form $P(c'|c)=P(c' \cap c)/P(c)$. Then, considering their independence, $P(c'|c)=(P(c')P(c))/P(c) = P(c')$. The same applies for $P(c'')$. The overall probability of the choice of a node and its children is then $P(c)P(c')P(c'') = \frac{m}{q}\frac{m'}{q'}\frac{m''}{q''} = \frac{m'm''}{q}$.

Finally, considering the chain of probabilities described by an $\textit{lmd}$, since for $\bar{A}$ $q = |L_{\bar{n}}(G)|$ and for each terminal production $p \in P_{\Sigma}$ $m = 1$, then for each $H \in L_{\bar{n}}(G)$ we can define its probability $P(H)$ to be generated as the productory of independent choices:

\begin{equation*}
\label{eq:UniformProbability}
P(H)={\underset{j = 1}{\overset{k}{\prod}}} P(c_j) = \frac{m_1}{|L_{\bar{n}}(G)|} \cdot \frac{m_2}{q_2} \cdot \frac{m_{k-1}}{q_{k-1}} \ldots \frac{1}{q_k}  = \frac{1}{|L_{\bar{n}}(G)|}
\end{equation*}

Each hypergraph $H \in L_{\bar{n}}(G)$ is generated over a uniform distribution given the uniformity of the sampling of the underlying $\textit{RNG}$.

\end{proof}

For the complexity analysis we consider the time required by the algorithm $\textbf{Gen}$ (Alg. \ref{alg:GenerationPhase}) for the generation of the hypergraph and the space required by the algorithm $\textbf{Pre}$ (Alg. \ref{alg:PreprocessingPhase}) to store the required data, taking into account that the input grammar is already provided in the correct $\textit{CNF}$ and the query to the $\textit{RNG}$ and the replacement operations are performed in unit time. The gaps present in the tables, that are not encountered in string method, are due to the possibility of a production to increase the size of the resulting hypergraph by more than $1$ in a single step.

\begin{theorem}\label{the:TimeComplexity}
With the assumptions of Theorem \ref{the:UniformSampling}, the size-$n$-hypergraph $H$ is generated by $\textbf{Gen}$ (Alg. \ref{alg:GenerationPhase}) in time $O(n^2)$.
\end{theorem}

\begin{proof}
The proof of Theorem \ref{the:TimeComplexity} is based on the analysis of the following recurrence relation for the function $\textbf{derH}$: $T(n) \leq cn + \underset{1 \leq k < (n - i)}{\textit{max}}[T(k) + T(n - k - i)]$, where $T(k)$ and $T(n - k - i)$ are the computational steps required to process the result of the split and $i$ is the number of internal nodes of the current production. In the worst case, we consider that $i = 0$ and that $k = 1$. A simple example is the discrete hypergraph language in which every iteration may generate a terminal hyperedge from $e_{\alpha}$ and the rest of the resulting hypergraph from $e_{\beta}$ without adding any node.
Since the choice of the production is constant, while the choice of a split is linear in $n$, choosing a split $n$ times leads to a quadratic behavior.

Since $i \ll n$, we may rewrite the recursion as:

\begin{equation*}
T(n) \leq cn + \underset{1 \leq k < n}{\textit{max}}[T(k) + T(n - k)]
\end{equation*}

Then, considering the worst case $k = 1$, for the next step of the recursion we obtain:

\begin{equation*}
T(n - 1) \leq c(n - 1) + \underset{1 \leq k < (n - 1)}{\textit{max}}[T(k) + T(n - k - 1)]
\end{equation*}

That is, at each step the choice of a split happens on an input of size $n - 1$. Since this choice requires linear time and it is taken $n$ times, the relation has solution $O(n^2)$.
\end{proof}

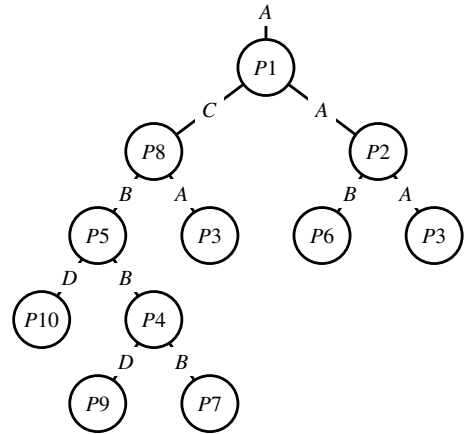
\begin{figure}[htpb]
\centering
\resizebox{6.0cm}{!}{%
\begin{tikzpicture}
 \begin{scope} [local bounding box = c1, shift = {(0.0, 0.0)}]
  \draw [ue] (0.0, 1.0) -- (0.0, 0.0) ;
  \draw [ue] (0.0, 0.0) -- (-2.0, -1.5) node [midway, te] {$C$} ;
  \draw [ue] (0.0, 0.0) -- (2.0, -1.5) node [midway, te] {$A$} ;
  \draw [ue] (-2.0, -1.5) -- (-3.0, -3.0) node [midway, te] {$B$} ;
  \draw [ue] (-2.0, -1.5) -- (-1.0, -3.0) node [midway, te] {$A$} ;
  \draw [ue] (2.0, -1.5) -- (1.0, -3.0) node [midway, te] {$B$} ;
  \draw [ue] (2.0, -1.5) -- (3.0, -3.0) node [midway, te] {$A$} ;
  \draw [ue] (-3.0, -3.0) -- (-4.0, -4.5) node [midway, te] {$D$} ;
  \draw [ue] (-3.0, -3.0) -- (-2.0, -4.5) node [midway, te] {$B$} ;
  \draw [ue] (-2.0, -4.5) -- (-3.0, -6.0) node [midway, te] {$D$} ;
  \draw [ue] (-2.0, -4.5) -- (-1.0, -6.0) node [midway, te] {$B$} ;
  \path (0.0, 0.0) pic{cn} node[] {$P1$} ;
  \path (-2.0, -1.5) pic{cn} node[] {$P8$} ;
  \path (2.0, -1.5) pic{cn} node[] {$P2$} ;
  \path (-3.0, -3.0) pic{cn} node[] {$P5$} ;
  \path (-1.0, -3.0) pic{cn} node[] {$P3$} ;
  \path (1.0, -3.0) pic{cn} node[] {$P6$} ;
  \path (3.0, -3.0) pic{cn} node[] {$P3$} ;
  \path (-4.0, -4.5) pic{cn} node[] {$P10$} ;
  \path (-2.0, -4.5) pic{cn} node[] {$P4$} ;
  \path (-3.0, -6.0) pic{cn} node[] {$P9$} ;
  \path (-1.0, -6.0) pic{cn} node[] {$P7$} ;
  \node [te] at (0.0, 1.0) {$A$} ;
 \end{scope}
\end{tikzpicture}%
}
\caption{Ordered tree $t$ for the derivation $d$ in Figure \ref{fig:termsgraphderivation}}
\label{fig:tree}
\end{figure}

We omit a discussion of the time complexity of the pre-processing phase (Alg. \ref{alg:PreprocessingPhase}) which can be shown to be linear, considering that given a grammar $G$ in $\textit{CNF}$, being its size $|G|$ constant, for each production $p$ a short form containing the information about the labels and the internal nodes is obtained in constant time.

\section{Conclusion}
\label{sec:conclusion}
Our main results, presented in Section \ref{sec:uniformcomplexity}, are that the method generates hypergraphs uniformly at random and in quadratic time. A topic for future work is to design an alternative generation algorithm that runs in linear time and quadratic space, following Mairson's second method in \cite{mairson-1994-gwi}. 

Another interesting topic is to extend the quasi-polynomial-time approximation algorithm of Gore et al.\ \cite{gore-1997-aqp} from strings to hypergraphs. This algorithm guarantees an approximated uniform distribution even for ambiguous grammars.

Our method allows to generate strings uniformly at random in some non-context-free string languages because hyperedge replacement grammars can specify certain \emph{string graph}\/ languages that are  not context-free. For example, this applies to the language $\{a^nb^nc^n \mid n \geq 0\}$. Moreover, our method is able to generate strings uniformly at random for a range of inherently ambiguous context-free languages.

The practically most promising application of our generation approach is the testing of programs in arbitrary programming languages that work on graphs. If the inputs of such programs are graphs in a context-free graph language, our method can generate test graphs uniformly at random in the domain of interest.
This should allow to refine random testing approaches such as \cite{chen-2010-art,hamlet-2002-rt0}.

\bibliography{final}
\end{document}